\documentclass[twocolumn]{revtex4}
\usepackage{amsthm, amssymb, graphicx, verbatim, tikz}

\usetikzlibrary{arrows,decorations.pathmorphing}

\def\t{^{\mbox{\tiny T}}}
\def\s{{\cal S}}
\def\g{{\cal G}}
\def\z{{\bf 0}}
\def\id{\mathbb{I}}
\theoremstyle{plain} 
\newtheorem{theorem}{Theorem}
\newtheorem{lemma}{Lemma}
\newtheorem{requirement}{Requirement}

\definecolor{darkgreen}{rgb}{0,0.5,0}

\begin{document}

\title{A derivation of quantum theory from physical requirements}
\author{Llu\'\i s Masanes}
\affiliation{ICFO-Institut de Ciencies Fotoniques, Mediterranean Technology Park, 08860 Castelldefels (Barcelona), Spain}
\author{Markus P.\ M\"uller}
\affiliation{Institute of Mathematics, Technical University of Berlin, 10623 Berlin, Germany,}
\affiliation{Institute of Physics and Astronomy, University of Potsdam, 14476 Potsdam, Germany}
\affiliation{Perimeter Institute for Theoretical Physics, 31 Caroline Street North, Waterloo, ON N2L 2Y5, Canada.}

\begin{abstract}
Quantum theory is usually formulated in terms of abstract mathematical postulates, involving Hilbert spaces, state vectors, and unitary operators. In this work, we show that the full formalism of quantum theory can instead be derived from five simple physical requirements, based on elementary assumptions about preparation, transformations and measurements. 
This is more similar to the usual formulation of special relativity, where two simple physical requirements -- the principles of relativity and light speed invariance -- are used to derive the mathematical structure of Minkowski space-time. Our derivation provides insights into the physical origin of the structure of quantum state spaces (including a group-theoretic explanation of the Bloch ball and its three-dimensionality), and it suggests several natural possibilities to construct consistent modifications of quantum theory.
\end{abstract}

\date{October 5, 2010}

\maketitle

\section{Introduction}

Quantum theory is usually formulated by postulating the mathematical structure and representation of states, transformations, and measurements. The general physical consequences that follow (like violation of Bell-type inequalities~\cite{J.Bell}, the possibility of performing state tomography with local measurements, or factorization of integers in polynomial time~\cite{n-factoring}) come as theorems which use the postulates as premises. In this work, this procedure is reversed: we impose five simple physical requirements, and this suffices to single out quantum theory and derive its mathematical formalism uniquely. This is more similar to the usual formulation of special relativity, where two simple physical requirements ---the principles of relativity and light speed invariance--- are used to derive the mathematical structure of Minkowski space-time and its transformations.

The requirements can be schematically stated as:
\begin{enumerate}

\item In systems that carry one bit of information, each state is characterized by a finite set of outcome probabilities. 

\item The state of a composite system is characterized by the statistics of measurements on the individual components. 

\item All systems that effectively carry the same amount of information have equivalent state spaces.

\item Any pure state of a system can be reversibly transformed into any other. 

\item In systems that carry one bit of information, all mathematically well-defined measurements are allowed by the theory.

\end{enumerate}
These requirements are imposed on the framework of generalized probabilistic theories \cite{tomo, Hardy5, Barn, Barr, gwmackey, BvN, darian}, which already assumes that some operational notions (preparation, mixture, measurement, and counting relative frequencies of measurement outcomes) make sense. Due to its conceptual simplicity, this framework leaves room for an infinitude of possible theories, allowing for weaker- or stronger-than-quantum non-locality~\cite{Barr, vandam, ic, pr, nav, boxworld}. In this work, we show that quantum theory (QT) and classical probability theory (CPT) are very special among those theories: they are the only general probabilistic theories that satisfy the five requirements stated above. 

The non-uniqueness of the solution is not a problem, since CPT is embedded in QT, thus QT is the most general theory satisfying the requirements. One can also proceed as Hardy in \cite{Hardy5}: if Requirement \ref{a.symmetry} is strengthened by imposing continuity of the reversible transformations, then CPT is ruled out and QT is the only theory satisfying the requirements. This strengthening can be justified by the continuity of time evolution of physical systems. 

It is conceivable that in the future, another theory may replace or generalize QT. Such a theory must violate at least one of our assumptions. The clear meaning of our requirements allows to straightforwardly explore potential features of such a theory. The relaxation of each of our requirements constitutes a different way to go beyond QT. 

The search for alternative axiomatizations of quantum theory (QT) is an old topic that goes back to Birkhoff and von Neumann~\cite{BvN}, and has been approached in many different ways: extending propositional logic \cite{BvN, gwmackey}, using operational primitives \cite{Hardy5, tomo, Barr, Barn, darian}, searching for information-theoretic principles \cite{Barr, Barn, Daki, bFuchs, br, vandam, ic}, building upon the phenomenon of quantum nonlocality \cite{pr, Barr, vandam, ic, nav}. Alfsen and Shultz~\cite{AlfsenShultzBook} have accomplished a complete characterization of the state spaces of QT from a geometric point of view, but the result does not seem to have an immediate physical meaning. In particular, the fact that the state space of a generalized bit is a three-dimensional ball is an assumption there, while here it is derived from physical requirements.

This work is particularly close to \cite{Hardy5, Daki}, from where it takes some material. More concretely, the multiplicativity of capacities and the Simplicity Axiom from \cite{Hardy5} are replaced by Requirement \ref{a.effects}. In comparison with~\cite{Daki}, the fact that each state of a generalized bit is the mixture of two distinguishable ones, the maximality of the group of reversible transformations and its orthogonality, and the multiplicativity of capacities, are also replaced by Requirement \ref{a.effects}.

{\em Summary of the paper.} Section II contains an introduction to the framework of generalized probabilistic theories, where some elementary results are stated without proof. In Section III the five requirements and their significance are explained in full detail. Section~\ref{sqt} is the core of this work. It contains the characterization of all theories compatible with the requirements, concluding that the only possibilities are CPT and QT. The Conclusion (Section \ref{s.conclusions}) recapitulates the results and adds some remarks. The Appendix contains all lemmas and their proofs.

\section{Generalized probabilistic theories}\label{gpt}

In CPT there can always be a joint probability distribution for all random variables under consideration. The framework of generalized probabilistic theories (GPTs), also called convex operational framework, generalizes this by allowing the possibility of random variables that cannot have a joint probability distribution, or cannot be simultaneously measured (like noncommuting observables in QT). 

This framework assumes that at some level there is a classical reality, where it makes sense to talk about experimentalists performing basic operations such as: preparations, mixtures, measurements, and counting relative frequencies of outcomes. These are the primary concepts of this framework. It also  provides a unified way for all GPTs to represent states, transformations and measurements. A particular GPT specifies which of these are allowed, but it does not tell their correspondence with actual experimental setups. On its own, a GPT can still make nontrivial predictions like: the maximal violation of a Bell inequality \cite{J.Bell}, the complexity-theoretic computational power \cite{n-factoring,theoryspace}, and in general, all information-theoretic properties of the theory \cite{Barr}.

The framework of GPTs can be stated in different ways, but all lead to the same formalism \cite{tomo,Hardy5,Barn,Barr,gwmackey,BvN,darian}. This formalism is presented in this section at a very basic level, providing some elementary results without proofs.

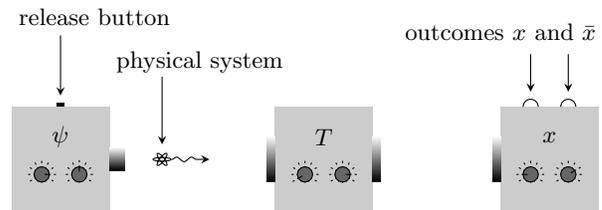
\begin{figure}
\begin{tikzpicture}[>=stealth]


\fill[fill=white!80!black] (.5,0) rectangle +(1.3,1.4); \shade[top
color=white, bottom color=black] (1.808,.55) rectangle +(.2,.3);
\draw (1.15,1) node {$\psi$};

\fill[fill=black] (1.1,1.4) rectangle +(.1,.05); \draw[<-]
(1.15,1.55) --+(0,.8); \draw (1.6,2.6) node {release button};

\draw [xshift=2.5cm, yshift=.7cm] (0,0) ellipse (.12cm and .025cm)
(0,0) [rotate=52] ellipse (.10cm and .03cm); \draw [xshift=2.5cm,
yshift=.7cm] (0,0) [rotate=128] ellipse (.10cm and .03cm); \draw
[->,decorate,
decoration={snake,amplitude=.3mm,segment length=2mm,post length=.8mm}] (2.63,.7) --+(.5,0);

\draw[<-] (2.5,.9) --+(0,.9); \draw (3,2) node {physical system};

\filldraw[fill=white!40!black, draw=black]  (.9,.5) circle (.1cm);
\draw (.9,.5)++(0:.04cm) --++(0:.08cm) ++(0:.04cm) --++(0:.03cm);
\draw (.9,.5)++(30:.1cm) --++(30:.01cm) ++(30:.04cm) --++(30:.03cm);
\draw (.9,.5)++(60:.1cm) --++(60:.01cm) ++(60:.04cm) --++(60:.03cm);
\draw (.9,.5)++(90:.1cm) --++(90:.01cm) ++(90:.04cm) --++(90:.03cm);
\draw (.9,.5)++(120:.1cm) --++(120:.01cm) ++(120:.04cm) --++(120:.03cm);
\draw (.9,.5)++(150:.1cm) --++(150:.01cm) ++(150:.04cm) --++(150:.03cm);
\draw (.9,.5)++(180:.1cm) --++(180:.01cm) ++(180:.04cm) --++(180:.03cm);
\draw (.9,.5)++(210:.1cm) --++(210:.01cm) ++(210:.04cm) --++(210:.03cm);
\draw (.9,.5)++(240:.1cm) --++(240:.01cm) ;
\draw (.9,.5)++(270:.1cm) --++(270:.01cm) ;
\draw (.9,.5)++(300:.1cm) --++(300:.01cm) ;
\draw (.9,.5)++(330:.1cm) --++(330:.01cm) ++(330:.04cm) --++(330:.03cm);

\filldraw[fill=white!40!black, draw=black]  (1.4,.5) circle (.1cm);
\draw (1.4,.5)++(0:.1cm) --++(0:.01cm) ++(0:.04cm) --++(0:.03cm);
\draw (1.4,.5)++(30:.1cm) --++(30:.01cm) ++(30:.04cm) --++(30:.03cm);
\draw (1.4,.5)++(60:.1cm) --++(60:.01cm) ++(60:.04cm) --++(60:.03cm);
\draw (1.4,.5)++(90:.04cm) --++(90:.08cm) ++(90:.04cm) --++(90:.03cm);
\draw (1.4,.5)++(120:.1cm) --++(120:.01cm) ++(120:.04cm) --++(120:.03cm);
\draw (1.4,.5)++(150:.1cm) --++(150:.01cm) ++(150:.04cm) --++(150:.03cm);
\draw (1.4,.5)++(180:.1cm) --++(180:.01cm) ++(180:.04cm) --++(180:.03cm);
\draw (1.4,.5)++(210:.1cm) --++(210:.01cm) ++(210:.04cm) --++(210:.03cm);
\draw (1.4,.5)++(240:.1cm) --++(240:.01cm) ;
\draw (1.4,.5)++(270:.1cm) --++(270:.01cm) ;
\draw (1.4,.5)++(300:.1cm) --++(300:.01cm) ;
\draw (1.4,.5)++(330:.1cm) --++(330:.01cm) ++(330:.04cm) --++(330:.03cm);


\fill[fill=white!80!black] (4,0) rectangle +(1.3,1.4);

\shade[top color=white, bottom color=black] (5.3,.4) rectangle
+(.1,.6); \shade[top color=white, bottom color=black] (4,.4)
rectangle +(-.1,.6);

\draw (4.65,1) node {$T$};
\filldraw[fill=white!40!black,draw=black]  (4.4,.5) circle (.1cm);
\draw (4.4,.5)++(0:.1cm) --++(0:.01cm) ++(0:.04cm) --++(0:.03cm);
\draw (4.4,.5)++(30:.1cm) --++(30:.01cm) ++(30:.04cm)
--++(30:.03cm); \draw (4.4,.5)++(60:.1cm) --++(60:.01cm)
++(60:.04cm) --++(60:.03cm); \draw (4.4,.5)++(90:.1cm)
--++(90:.01cm) ++(90:.04cm) --++(90:.03cm); \draw
(4.4,.5)++(120:.1cm) --++(120:.01cm) ++(120:.04cm)
--++(120:.03cm); \draw (4.4,.5)++(150:.1cm) --++(150:.01cm)
++(150:.04cm) --++(150:.03cm); \draw (4.4,.5)++(180:.1cm)
--++(180:.01cm) ++(180:.04cm) --++(180:.03cm); \draw
(4.4,.5)++(210:.04cm) --++(210:.08cm) ++(210:.04cm)
--++(210:.03cm); \draw (4.4,.5)++(240:.1cm) --++(240:.01cm) ;
\draw (4.4,.5)++(270:.1cm) --++(270:.01cm) ; \draw
(4.4,.5)++(300:.1cm) --++(300:.01cm) ; \draw (4.4,.5)++(330:.1cm)
--++(330:.01cm) ++(330:.04cm) --++(330:.03cm);

\filldraw[fill=white!40!black, draw=black]  (4.9,.5) circle
(.1cm); \draw (4.9,.5)++(0:.04cm) --++(0:.08cm) ++(0:.04cm)
--++(0:.03cm); \draw (4.9,.5)++(30:.1cm) --++(30:.01cm)
++(30:.04cm) --++(30:.03cm); \draw (4.9,.5)++(60:.1cm)
--++(60:.01cm) ++(60:.04cm) --++(60:.03cm); \draw
(4.9,.5)++(90:.1cm) --++(90:.01cm) ++(90:.04cm) --++(90:.03cm);
\draw (4.9,.5)++(120:.1cm) --++(120:.01cm) ++(120:.04cm)
--++(120:.03cm); \draw (4.9,.5)++(150:.1cm) --++(150:.01cm)
++(150:.04cm) --++(150:.03cm); \draw (4.9,.5)++(180:.1cm)
--++(180:.01cm) ++(180:.04cm) --++(180:.03cm); \draw
(4.9,.5)++(210:.1cm) --++(210:.01cm) ++(210:.04cm)
--++(210:.03cm); \draw (4.9,.5)++(240:.1cm) --++(240:.01cm) ;
\draw (4.9,.5)++(270:.1cm) --++(270:.01cm) ; \draw
(4.9,.5)++(300:.1cm) --++(300:.01cm) ; \draw (4.9,.5)++(330:.1cm)
--++(330:.01cm) ++(330:.04cm) --++(330:.03cm);


\fill[fill=white!80!black] (7,0) rectangle +(1.3,1.4); \shade[top
color=white, bottom color=black] (7,.4) rectangle +(-.1,.6);

\draw[<-] (7.4,1.6) --+(0,.5); \draw[<-] (7.9,1.6) --+(0,.5);
\draw (7,2.4) node {outcomes $x$ and $\bar{x}$};

\draw (7.5,1.4) arc (0:180:.1cm); \draw (8,1.4) arc (0:180:.1cm);

\draw (7.65,1) node {$x$};
\filldraw[fill=white!40!black,draw=black]  (7.4,.5) circle (.1cm);

\draw (7.4,.5)++(0:.1cm) --++(0:.01cm) ++(0:.04cm) --++(0:.03cm);

\draw (7.4,.5)++(30:.1cm) --++(30:.01cm) ++(30:.04cm)
--++(30:.03cm);

\draw (7.4,.5)++(60:.1cm) --++(60:.01cm) ++(60:.04cm)
--++(60:.03cm);

\draw (7.4,.5)++(90:.1cm) --++(90:.01cm) ++(90:.04cm)
--++(90:.03cm);

\draw (7.4,.5)++(120:.1cm) --++(120:.01cm) ++(120:.04cm)
--++(120:.03cm);

\draw (7.4,.5)++(150:.1cm) --++(150:.01cm) ++(150:.04cm)
--++(150:.03cm);

\draw (7.4,.5)++(180:.04cm) --++(180:.08cm) ++(180:.04cm)
--++(180:.03cm);

\draw (7.4,.5)++(210:.1cm) --++(210:.01cm) ++(210:.04cm)
--++(210:.03cm);

\draw (7.4,.5)++(240:.1cm) --++(240:.01cm);

\draw (7.4,.5)++(270:.1cm) --++(270:.01cm);

\draw (7.4,.5)++(300:.1cm) --++(300:.01cm);

\draw (7.4,.5)++(330:.1cm) --++(330:.01cm) ++(330:.04cm)
--++(330:.03cm);

\filldraw[fill=white!40!black, draw=black]  (7.9,.5) circle
(.1cm);

\draw (7.9,.5)++(0:.1cm) --++(0:.01cm) ++(0:.04cm) --++(0:.03cm);

\draw (7.9,.5)++(30:.04cm) --++(30:.08cm) ++(30:.04cm)
--++(30:.03cm);

\draw (7.9,.5)++(60:.1cm) --++(60:.01cm) ++(60:.04cm)
--++(60:.03cm);

\draw (7.9,.5)++(90:.1cm) --++(90:.01cm) ++(90:.04cm)
--++(90:.03cm);

\draw (7.9,.5)++(120:.1cm) --++(120:.01cm) ++(120:.04cm)
--++(120:.03cm); \draw (7.9,.5)++(150:.1cm) --++(150:.01cm)
++(150:.04cm) --++(150:.03cm); \draw (7.9,.5)++(180:.1cm)
--++(180:.01cm) ++(180:.04cm) --++(180:.03cm); \draw
(7.9,.5)++(210:.1cm) --++(210:.01cm) ++(210:.04cm)
--++(210:.03cm); \draw (7.9,.5)++(240:.1cm) --++(240:.01cm) ;
\draw (7.9,.5)++(270:.1cm) --++(270:.01cm) ; \draw
(7.9,.5)++(300:.1cm) --++(300:.01cm) ; \draw (7.9,.5)++(330:.1cm)
--++(330:.01cm) ++(330:.04cm) --++(330:.03cm);

\end{tikzpicture}
\caption{{\bf General experimental set up.} From left to right there are the preparation, transformation and measurement devices. As
soon as the release button is pressed, the preparation device outputs a physical system in the state specified by the knobs. The next device performs the transformation specified by its knobs (which in particular can be ``do nothing"). The device on the right performs the measurement specified by its knobs, and the outcome \mbox{($x$ or $\bar{x}$)} is indicated by the corresponding light.
\label{f11}}
\end{figure}


\subsection{States}\label{stmeas1}

{\em Definition of system.} To a setup like FIG. 1 we associate a system if for each configuration of the preparation, transformation and measurement devices, the relative frequencies of the outcomes tend to a unique probability distribution (in the large sample limit). 

The probability of a measurement outcome $x$ is denoted by $p(x)$. This outcome can be associated to a binary measurement which tells whether $x$ happens or not (this second event $\bar x$ has probability $p(\bar x) =1-p(x)$). The above definition of system allows to associate to each preparation procedure a list of probabilities for the outcomes of all measurements that can be performed on a system. As we show in Subsection~\ref{SubsecCD} below, our requirements imply that all these probabilities $p(x)$ are determined by a finite set of them; the smallest such set is used to represent the state
\begin{equation}\label{psi state}
	\psi= \left[ \begin{array}{c}
		1 \\ p(x_1) \\ \vdots \\ p(x_d)
	\end{array} \right]=
	\left[\begin{array}{c}
	   \psi^0  \\ \psi^1 \\ \vdots \\ \psi^d
	\end{array}
	\right]
	\in\ \s\ \subset\ \mathbb{R}^{d+1}.
\end{equation}
The measurement outcomes that characterize the state $x_1,\ldots,x_d$ are called {\em fiducial}, and in general, there is more than one set of them (for example, a $\frac 1 2$-spin particle in QT is characterized by the spin in any 3 linearly-independent directions). Note that each of the fiducial outcomes can correspond to a different measurement. The redundant component $\psi^0 =1$ is reminiscent of QT, where one of the diagonal entries of a density matrix is redundant, since they sum up to $1$. In fact $\psi^0\neq 1$ is sometimes used to represent unnormalized states, but not in this paper, where only normalized states are considered. The redundant component $\psi^0$ allows to use the tensor-product formalism in composite systems (Subsection \ref{multip_ns}), which simplifies the notation.

The set of all allowed states $\s$ is convex \cite{convex_book}, because if $\psi_1, \psi_2 \in\s$ then one can prepare $\psi_1$ with probability $q$ and $\psi_2$ with probability $1-q$, effectively preparing the state $q\psi_1 + (1-q)\psi_2$. The number of fiducial probabilities $d$ is equal to the (affine) dimension of $\s$, otherwise one fiducial probability would be functionally related to the others, and hence redundant.

Suppose there is a $\mathbb{R}^{d+1}$-vector $\psi \notin \s$ which is in the topological closure of $\s$ -- that is, $\psi$ can be approximated by states $\psi'\in\s$ to arbitrary accuracy. Since there is no observable physical difference between \emph{perfect preparation} and \emph{arbitrarily good preparation}, we will consider $\psi$ to be a valid state and add it to the state space. This does not change the physical predictions of the theory, but it has the mathematical consequence that state spaces become topologically closed. Since state vectors~(\ref{psi state}) are bounded, and we are in finite dimensions (shown in Subsection~\ref{SubsecCD}), state spaces $\s$ are \emph{compact} convex sets~\cite{convex_book}.

The {\em pure states} of a state space $\s$ are the ones that cannot be written as mixtures: $\psi \neq q\psi_1 +(1-q)\psi_2$ with $\psi_1 \neq \psi_2$ and $0<q<1$. Since $\s$ is compact and convex, all states are mixtures of pure states  \cite{convex_book}.

\subsection{Measurements}\label{stmeas2}

The probability of measurement outcome $x$ when the system is in state $\psi\in\s$ is given by a function $\Omega_x(\psi)$. Suppose the system is prepared in the mixture \mbox{$q\psi_1 +(1-q) \psi_2$}, then the relative frequency of outcome $x$ does not depend on whether the label of the actual preparation $\psi_k$ is ignored before or after the measurement, hence
\[
	\Omega_x \!\big( q\psi_1 +(1-q) \psi_2 \big) = 
	q\, \Omega_x (\psi_1) +(1-q)\, \Omega_x (\psi_2)\ .
\]
This means that the function $\Omega_x$ is affine on $\s$. The redundant component $\psi^0$ in (\ref{psi state}) allows to write this function as a linear map $\Omega_x: \mathbb{R}^{d+1} \to \mathbb{R}$ \cite{tomo,Barr}.

An \emph{effect} is a linear map $\Omega: \mathbb{R}^{d+1} \to\mathbb{R}$ such that $\Omega(\psi) \in [0,1]$ for all states $\psi\in\s$. Every function $\Omega_x$ associated to an outcome probability $p(x)$ is an effect. The converse is not necessarily true: the framework of GPTs allows to construct theories where some effects do not represent possible measurement outcomes. These restrictions are analogous to superselection rules, where some (mathematically well-defined) states are not allowed by the physical theory. This is related to Requirement \ref{a.effects}.
A {\em tight effect} $\Omega$ is one for which there are two states $\psi_0,\psi_1 \in\s$ satisfying $\Omega(\psi_0)=0$ and $\Omega(\psi_1)=1$. 

An $n$-outcome measurement is specified by $n$ effects $\Omega_1,\ldots, \Omega_n$ such that $\Omega_1(\psi)+\cdots +\Omega_n(\psi) =1$ for all $\psi\in\s$. The number $\Omega_a (\psi)$ is the probability of outcome $a$ when the measurement is performed on the state $\psi$. The states $\psi_1,\ldots, \psi_n$ are {\em distinguishable} if there is an $n$-outcome measurement such that $\Omega_a (\psi_{b}) = \delta_{a,b}$, where
$\delta_{a,b}=1$ if $a=b$, and $\delta_{a,b}=0$ if $a\neq b$.

The {\em capacity} of a state space $\s$ is the size of the largest family of distinguishable states, and is denoted by $c$. This is the amount of classical information that can be transmitted by the corresponding type of system, in a single-shot error-free procedure. (In QT the capacity of a system is the dimension of its corresponding Hilbert space; which must not be confused with the dimension of the state space $d= c^2-1$, that is, the set of $c\times c$ complex matrices that are positive and have unit trace.) A {\em complete measurement} on $\s$ is one capable of distinguishing $c$ states.

\subsection{Transformations}\label{s.trans}

Each type of system has associated to it: a state space, a set of measurements, and a set of transformations. A transformation $T$ is a map $T: \s \to\s$. Similarly as for measurements, if a state is prepared as a mixture $q\psi_1+(1-q)\psi_2$, it does not matter whether the label of the actual preparation $\psi_k$ is ignored before or after the transformation. Hence
\[
   T\left(q\psi_1+(1-q)\psi_2\right)=q T(\psi_1)+(1-q) T(\psi_2)\ ,
\]
which implies that $T$ is an affine map. The redundant component $\psi^0$ in (\ref{psi state}) allows to extend $T$ to a linear map $T: \mathbb{R}^{d+1} \to \mathbb{R}^{d+1}$ \cite{tomo,Barr}.

A transformation $T$ is reversible if its inverse $T^{-1}$ exists and belongs to the set of transformations allowed by the theory. The set of (allowed) reversible transformations of a particular state space $\s$ forms a group $\g$. For the same reason as for the state space itself, we will assume that the group of reversible transformations is topologically closed. Previously we have seen that a state space $\s$ is bounded, hence the corresponding group of transformations $\g$ is bounded, too. In summary, groups of transformations are compact \cite{matrix_groups}.

\subsection{Composite systems}\label{multip_ns}

{\em Definition of composite system.} Two systems $A,B$ constitute a composite system, denoted $AB$, if a measurement for $A$ together with a measurement for $B$ uniquely specifies a measurement for $AB$. This means that
if $x$ and $y$ are measurement outcomes on $A$ and $B$ respectively, the pair $(x,y)$ specifies a unique measurement outcome on $AB$, whose probability distribution $p(x,y)$ does not depend on the temporal order in which the subsystems are measured.

The fact that subsystems are themselves systems implies that each has a well-defined reduced state $\psi_{A}, \psi_B$ which does not depend on which transformations and measurements are performed on the other subsystem (see definition of system in Subsection \ref{stmeas1}). This is often referred to as no-signaling. Let $x_1,\ldots, x_{d_A}$ be the fiducial measurements of system $A$, and $y_1,\ldots, y_{d_B}$ the ones of $B$. The no-signaling constraints are
\begin{equation} \label{nosig}
\begin{array}{lll}
	p(x_i) &=& p(x_i,y_j) +p(x_i,\bar{y}_j)\\
	p(y_i) &=& p(x_i,y_j) +p(\bar x_i,y_j)
\end{array}
\end{equation}
for all $i,j$. 

An assumption which is often postulated additionally in the GPT context is Requirement~\ref{a.correlations}, which says
that the state of a composite system is completely characterized by the statistics of measurements on the subsystems, that is, $p(x,y)$. This and no-signaling (\ref{nosig}) imply that states in $AB$ can be represented on the tensor product vector space~\cite{tomo} as
\begin{equation}\label{psiAB}
	\psi_{AB} = \left( \begin{array}{c} 
	1 \\ \vdots\\ p(x_i)\\ \vdots\\ p(y_j)\\ \vdots\\ p(x_i, y_j)\\ \vdots 
	\end{array} \right)\ \in\ \s_{AB} \subset \mathbb{R}^{d_A+1}\! \otimes
\mathbb{R}^{d_B +1} .
\end{equation}
The joint probability of two arbitrary local measurement outcomes $x,y$ is given by
\begin{equation}\label{localmeas}
	p(x,y)= (\Omega_x \otimes \Omega_y) (\psi_{AB})\ ,
\end{equation}
where $\Omega_x$ is the effect representing $x$ in $A$, that is $p(x)= \Omega_x (\psi_A)$, and analogously for $\Omega_y$ \cite{tomo}. (The term ``local" is used when referring to subsystems, and has nothing to do with spatial locations.) In other words, if $\{ \Omega^A_1, \ldots, \Omega^A_n \}$ is an $n$-outcome measurement on $A$, and $\{ \Omega^B_1, \ldots, \Omega^B_m \}$ is an $m$-outcome measurement on $B$, then \mbox{$\{ \Omega^A_a \otimes \Omega^B_b |$ $ a=1,\ldots,n\, ; b=1,\ldots,m\}$} defines a measurement on $AB$ with $nm$ outcomes.
Local transformations act on the global state as
\begin{equation}\label{localtrans}
	\psi_{AB}\ \rightarrow\  
	(T_A \otimes T_B)(\psi_{AB})\ ,
\end{equation}
where $T_A$ is the matrix that represents the transformation in $A$, and analogously for $T_B$ \cite{tomo}. The reduced states
\begin{equation}\label{red-state}
	\psi_A = \left( \begin{array}{c} 
	1 \\ \vdots\\ p(x_i)\\ \vdots
	\end{array} \right)\ ,\quad
	\psi_B = \left( \begin{array}{c} 
	1 \\ \vdots\\ p(y_j)\\ \vdots
	\end{array} \right)\ ,
\end{equation}
are obtained from $\psi_{AB}$ by picking the right components (\ref{psiAB}). Alternatively, reduced states can be defined by $\Omega_A(\psi_A)= (\Omega_A \otimes \mathbf{1}) (\psi_{AB})$ for any effect $\Omega_A$ in $A$, where $\mathbf{1} (\psi_B) =\psi_B^0$ is the unit effect. The reduced state $\psi_A$ must belong to the state space of subsystem $A$, denoted $\s_A$, and any state in $\s_A$ must be the reduction of a state from $\s_{AB}$. (Analogously for subsystem B.) This implies that all product states
\begin{equation}\label{prod-state}
	\psi_{AB}= \psi_A \otimes \psi_B
\end{equation}
are contained in $\s_{AB}$ \cite{tomo}, and similarly, all tensor products of local measurements and transformations are allowed on $AB$.

Given two fixed state spaces $\s_A$ and $\s_B$, the previous discussion imposes constraints on the state space of the composite system $\s_{AB}$. However, there are still many different possible joint state spaces $\s_{AB}$, and some of them allow for larger violations of Bell inequalities than QT. In fact, this has been extensively studied \cite{nav,vandam,boxworld,ic,pr,Barn,Barr}, and is one of the reasons for the popularity of generalized probabilistic theories.

Nothing prevents Bob's system from being composite itself; hence one can recursively extend the definition of composite system and formulas (\ref{psiAB}), (\ref{localmeas}), (\ref{localtrans}), and (\ref{prod-state}) to more parties. 

\subsection{Equivalent state spaces} \label{equivss}

Let ${\cal L}: \s \to \s'$ be an invertible affine map. If all states are transformed as $\psi \rightarrow {\cal L} (\psi)$, and all effects on $\s$ are transformed as $\Omega \rightarrow \Omega\circ {\cal L}^{-1}$, then the outcome probabilities $\Omega(\psi)$ are kept unchanged. Analogously, if all transformations on $\s$ are mapped as $T \rightarrow {\cal L}\circ T\circ {\cal L}^{-1}$ then their action on the states is the same. The new state space $\s'$, together with the transformed effects and transformations, is then just a different representation of $\s$. In this case, we call $\s$ and $\s'$ \emph{equivalent}.
In the new representation, the entries of $\psi$ need not be probabilities as in~(\ref{psi state}), but it may have other advantages. In this work, several representations are used.

In the standard formalism of QT, states are represented by density matrices, however they can also be represented as in (\ref{psi state}). 

Changing the set of fiducial measurements is a particular type of ${\cal L}$-transformation. For example, if the components of the Bloch vector (of a quantum \mbox{spin-$\frac 1 2$} particle) correspond to spin measurements in non-orthogonal directions, then the Bloch sphere becomes an ellipsoid.

\subsection{Instances of generalized probabilistic theories}\label{s.instances}

QT is an instance of GPT, and can be specified as follows. The state space $\s_c$ with capacity $c$ is equivalent to the set of complex $c\!\times\! c$-matrices $\rho$ such that $\rho \geq 0$ and ${\rm tr} \rho =1$. This set has dimension $d_c =$ \mbox{$c^2 -1$}, and its pure states are rank-one. The effects on $\s_c$ have the form $\Omega(\rho)= \mathrm{tr} (M\rho)$, where $M$ is a complex $c\! \times\! c$-matrix such that $0\leq M \leq \mathbb{I}$. The reversible transformations act as $\rho \rightarrow V\rho V^\dagger$ with $V\in \mathrm{SU}(c)$. The capacity of a composite system $AB$ is the product of the capacities for the subsystems $c_{AB} = c_A c_B$.

CPT is another instance of GPT, and can be specified as follows. The state space $\s_c$  with capacity $c$ is equivalent to the set of $c$-outcome probability distributions $[p(1), \ldots, p(c)]$, which has dimension $d_c =c-1$ (in geometric terms, each $\s_c$ is a simplex). The pure states are the deterministic distributions $p(a) = \delta_{a,b}$ with $b=1,\ldots, c$. The $c$-outcome measurement with effects $\Omega_a (\psi) =p(a)$ for $a=1,\ldots, c$, distinguishes the $c$ pure states, hence it is complete. Any other measurement is a function of this one. The reversible transformations act by permuting the entries of the state $[p(1), \ldots, p(c)]$. The capacity of a composite system is also $c_{AB} = c_A c_B$. Note that CPT can be obtained by restricting the states of QT to diagonal matrices. In other words, CPT is embedded in QT.

An instance of GPT that is not observed in nature is \emph{generalized no-signaling theory}~\cite{Barr}, colloquially called \emph{boxworld}. By definition, state spaces contain all correlations~(\ref{psiAB}) satisfying the no-signaling constraints~(\ref{nosig}). Such state spaces have finitely many pure states, and some of them violate Bell inequalities stronger than any quantum state~\cite{pr}. The effects in boxworld are all generated by products of local effects. The group of reversible transformations consists only of relabellings of local measurements and their outcomes, permutations of subsystems, and combinations thereof~\cite{boxworld}.

\section{The requirements} 

This section contains the precise statement of the requirements, each followed by explanations about its significance. 

\begin{requirement}[Finiteness]\label{a.finiteness} 
A state space with capacity $c=2$ has finite dimension $d$.
\end{requirement}

If this did not hold, the characterization of a state of a generalized bit would require infinitely many outcome probabilities, making state estimation impossible. It is shown below that this requirement, together with the others, implies that all state spaces with finite capacity $c$ have finite dimension.

\begin{requirement} [Local tomography] \label{a.correlations} 
	The state of a composite system $AB$ is completely characterized by the statistics of measurements on the subsystems $A,B$.
\end{requirement}

In other words, state tomography \cite{tomo} can be performed locally. This is equivalent to the constraint 
\begin{equation}\label{h dim}
	(d_{AB}+1)= (d_A+1) (d_B+1)
\end{equation}
\cite{Hardy5,tomo}. This requirement can be recursively extended to more parties by letting subsystems $A,B$ to be themselves composite. 

\begin{requirement}[Equivalence of subspaces] \label{a.subspaces} 
	Let $\s_c$ and $\s_{c-1}$ be systems with capacities $c$ and $c-1$, respectively. If $\Omega_1,\ldots,\Omega_c$ is a complete measurement on $\s_c$, then the set of states $\psi\in\s_c$ with $\Omega_c(\psi)=0$ is equivalent to $\s_{c-1}$.
\end{requirement}

The notions of complete measurements and equivalent state spaces are defined in Subsections \ref{stmeas2} and \ref{equivss}.
In particular, equivalence of $\s_{c-1}$ and 
\begin{equation}\label{S'}
	\s'_{c-1} := \{\psi\in \s_c :\Omega_c (\psi)= 0\} \subset \s_c 
\end{equation}
implies that all measurements and reversible transformations on one of them can be implemented on the other. 

This requirement, first introduced in~\cite{Hardy5}, implies that all state spaces with the same capacity are equivalent: if $\s_{c-1}$ and $\tilde\s_{c-1}$ are state spaces with capacity $c-1$, then both are equivalent to (\ref{S'}), hence they are equivalent to each other. In other words, the only property that characterizes the type of system is the capacity for carrying information. If we start with $\s_c$ and apply Requirement~\ref{a.subspaces} recursively, we get a more general formulation: consider any subset of outcomes $\{a_1,\ldots,a_{c'}\}\subseteq\{1,\ldots,c\}$ of the complete measurement $\Omega_1, \ldots, \Omega_c$, then the set of states $\psi\in\s_c$ with
\begin{equation}\label{ss}
	\Omega_{a_1}(\psi) +\cdots +\Omega_{a_{c'}}(\psi) =1
\end{equation}
is equivalent to the state space $\s_{c'}$ with capacity $c'$. This provides an onion-like structure for all state spaces \mbox{$\s_1 \subset \s_2 \subset \s_3 \subset \cdots$}

The particular structure of QT simplifies the task of assigning a state space to a physical system or experimental setup. It is not necessary to consider all possible states of the system, but instead, the relevant ones for the context being analyzed. For example, an atom is sometimes modeled with a state space having two distinguishable states ($c=2$), even though its constituents have many more degrees of freedom. In particular, if we know that only two energy levels are populated with non-zero probability, we can ignore all others and effectively get a genuine quantum 2-level state space. In a theory where this is not true, the effective state space might depend on how many unpopulated energy levels are ignored, or on the detailed internal state of the electron, for example. In order to avoid pathologies like this, we postulate Requirement~\ref{a.subspaces}.

\begin{requirement}[Symmetry]\label{a.symmetry}
	For every pair of pure states $\psi_1,\psi_2 \in \s$ there is a reversible transformation $G$ mapping one onto the other: $G(\psi_1) = \psi_2$.
\end{requirement}

The set of reversible transformations of a state space $\s_c$ forms a group, denoted $\g_c$. This group endows $\s_c$ with a symmetry, which makes all pure states equivalent. A group $\g_c$ is said to be {\em continuous} if it is topologically connected: any transformation is the composition of many infinitesimal ones \cite{matrix_groups}. Hardy invokes the continuity of time-evolution in physical systems to justify the continuity of reversible transformations \cite{Hardy5,tomo}; in this case, state spaces $\s_c$ must have infinitely-many pure states; this rules out CPT and singles out QT. However, all the analysis in this work is done without imposing continuity, since we find it very interesting that the only theory with state spaces having finitely-many pure states, and satisfying the requirements, is CPT.

\begin{requirement}[All measurements allowed] \label{a.effects} 
All effects on $\s_2$ are outcome probabilities of possible measurements.
\end{requirement}

It is shown below that, in combination with the other requirements, this implies that all effects on all state spaces (with arbitrary $c$) appear as outcome probabilities of measurements in the resulting theory. Note that Requirement~\ref{a.effects} has non-trivial consequences in conjunction with the other requirements: adding effects as allowed measurements to a physical theory extends the applicability of
Requirement~\ref{a.subspaces}.

For completeness, we would like to mention that Requirement~\ref{a.effects} can be replaced by the following postulate, which has first been put forward in an interesting paper that appeared after completion of this work~\cite{Giulio}. It calls a state ``completely mixed'' if it is in the relative interior of state space. See Lemma~\ref{LemFive} in the appendix for how the proof of our main result has to be modified in this case.

\bigskip\noindent
{\bf Requirement 5'~\cite{Giulio}.}  If a state is not completely mixed, then there exists at least one state that can be perfectly distinguished from it.

\section{Characterization of all theories satisfying the requirements}\label{sqt}

\subsection{The maximally-mixed state}

We use the following notation: the system with capacity $c$ has state space $\s_c$ with dimension $d_c$ and group of reversible transformations $\g_c$.
The group $\g_c$ is compact (Section \ref{s.trans}), and hence, has a normalized invariant Haar measure \cite{Haar-book}. This allows to define the maximally-mixed state
\begin{equation}\label{maxmix}
	\mu_c = \int_{\g_c}\!\! G(\psi)\, dG\ \in \s_c\ ,
\end{equation}
where $\psi\in \s_c$ is an arbitrary pure state. It follows from Requirement~\ref{a.symmetry} that the resulting state $\mu_c$ does not depend on the choice of the pure state $\psi$. By construction, the maximally-mixed state is invariant:
\begin{equation}\label{inv state}
	G(\mu_c) =\mu_c\ \mbox{ for all }\ G\in \g_c\ . 
\end{equation}
Moreover, Lemma~\ref{LemMaxMixUnique} shows that it is the only invariant state in $\s_c$ (this lemma and all others are stated and proven in the appendix).

\subsection{The generalized bit}\label{gener-bit}

A generalized bit is a system with capacity two. For any state $\psi\in \s_2$ in the standard representation (\ref{psi state}), its Bloch representation is defined by
\begin{equation}\label{hat psi}
	\hat\psi= 2 \left[ \begin{array}{c}
		p(x_1) - \mu^1_2 \\ \vdots \\ 
		p(x_{d_2}) - \mu_2^{d_2} 
	\end{array} \right]\
	\in \hat\s_2\ \subset \mathbb{R}^{d_2}\ .
\end{equation}
States in the Bloch representation do not have the redundant component $\psi^0$, so equations (\ref{localmeas}, \ref{localtrans}, \ref{prod-state}) become less simple. The invertible map ${\cal L}: \s_2 \rightarrow \hat\s_2$ is affine but not linear; hence, effects $\Omega$ in the Bloch representation ($\hat\Omega = \Omega \circ {\cal L}^{-1}$) are affine but not necessarily linear. The same applies to transformations $(\hat G= {\cal L}\circ G\circ {\cal L}^{-1})$, however, the maximally-mixed state in the Bloch representation is the null vector $\hat\mu_2 =\z$, therefore (\ref{inv state}) becomes $\hat G (\z) =\z$, which implies that $\hat G$ acts linearly (as a matrix).

\begin{theorem}\label{pure=boundary}
	A state in $\hat\s_2$ is pure if and only if it belongs to the boundary $\partial \hat\s_2$.
\end{theorem}

\begin{proof}
In any convex set, pure states belong to the boundary \cite{convex_book}. Let us see the converse.

It is shown in \cite{Strask} that any compact convex set has a supporting hyperplane containing exactly one point of the set. Translated to our language: there is a tight effect $\hat\Omega_\mathrm{one}$ on $\hat\s_2$ such that only one state $\hat\varphi_\mathrm{one} \in \hat\s_2$ satisfies $\hat\Omega_\mathrm{one} (\hat\varphi_\mathrm{one}) =1$; this is illustrated in FIG. \ref{f22}. According to Requirement~\ref{a.effects}, the effect $\hat\Omega_\mathrm{one}$ corresponds to a valid measurement outcome, and so does $\hat\mathbf{1} -\hat\Omega_\mathrm{one}$, where $\hat\mathbf{1} (\hat\psi)=1$ for all $\hat\psi \in \hat\s_2$. Thus, the two effects $\hat\Omega_\mathrm{one}$ and $\hat\mathbf{1} -\hat\Omega_\mathrm{one}$ define a complete measurement on $\hat\s_2$. Imposing Requirement \ref{a.subspaces} on the single outcome $\hat\Omega_\mathrm{one}$ constrains the state space with unit capacity $\hat\s_1$ to contain only one state.

Suppose there is a point in the boundary $\hat \varphi_\mathrm{mix} \in\partial \hat\s_2$ which is not pure: $\hat \varphi_\mathrm{mix}= q \hat \varphi_1 +(1-q) \hat \varphi_2$ with $\hat \varphi_1 \neq \hat \varphi_2$ and $0<q<1$. Every point in the boundary of a compact convex set has a supporting hyperplane which contains it \cite{convex_book}. In our language: there is a tight effect $\hat\Omega$ on $\hat\s_2$ such that $\hat \Omega (\hat \varphi_\mathrm{mix}) =1$. The affine function $\hat\Omega$ is bounded: $\hat\Omega (\hat \varphi) \leq 1$ for any $\hat \varphi \in\hat \s_2$, which implies $\hat\Omega (\hat \varphi_1)= \hat\Omega (\hat \varphi_2)= 1$; this is illustrated in FIG. \ref{f22}. Like $\hat\Omega_\mathrm{one}$, the effect $\hat\Omega$ defines a complete measurement, and Requirement \ref{a.subspaces} can be imposed on the single outcome $\hat\Omega$, implying that $\hat\s_1$ contains more than one state. This is in contradiction with the previous paragraph; hence, all points in the boundary are pure. 
\end{proof}

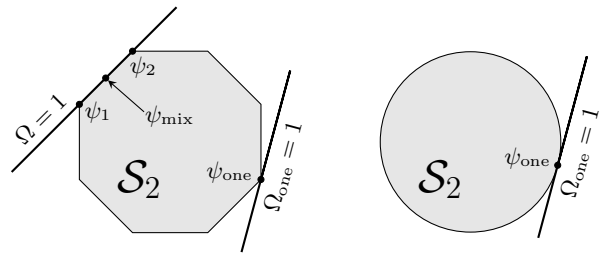
\begin{figure}

\begin{tikzpicture}[>=stealth]
	\filldraw[fill=white!90!black, draw=black] (0,0)--++(90:1cm)--++(45:1cm)--++(0:1cm)--++(-45:1cm) --++(-90:1cm)--++(-135:1cm)--++(-180:1cm) --cycle;

	\draw (.8,0) node {\LARGE ${\cal S}_2$};

	\draw[thick] (-.9,.1)--+(45:3.1cm) node[near start,sloped,above] {$\Omega =1$};
	\draw[fill=black] (0,1) circle (.04cm);
	\draw (0.25,0.9) node {$\psi_1$};

	\draw[fill=black] (.7071,1.7071) circle (.04cm);
	\draw (.85,1.5) node {$\psi_2$};

	\draw[fill=black] (.35,1.35) circle (.04cm);
	\draw (1.2,.85) node {$\psi_\mathrm{mix}$};
	\draw[->] (.86,.9) --(.38,1.32); 

	\draw[thick] (2.41421,0)--+(75:1.5cm)--+(255:1cm) node[midway,sloped,below] {$\Omega_\mathrm{one} =1$};
	\draw[fill=black] (2.41421,0) circle (.04cm);
	\draw (2,0.1) node {$\psi_\mathrm{one}$};

	\filldraw[fill=white!90!black, draw=black] (5.2,.5) circle (1.2cm);

	\draw (4.8,0) node {\LARGE ${\cal S}_2$};

	\draw[thick] (5.2,.5)++(-15:1.2cm) --+(75:1.5cm) --+(255:1cm) node[midway,sloped,below] {$\Omega_\mathrm{one} =1$};
	\draw[fill=black] (5.2,.5)++(-15:1.2cm) circle (.04cm);
	\draw (4.8,.6)++(-15:1.2cm) node {$\psi_\mathrm{one}$};

\end{tikzpicture}

\caption{The left figure is a state space whose boundary consists of facets (like $\Omega =1$). Each facet contains infinitely many states ($\Omega =1$ contains $\psi_1$, $\psi_2$ and all $\psi_\mathrm{mix} = q\psi_1 + (1-q) \psi_2$). The right figure is a state space whose boundary has no facets. Any state space has supporting hyperplanes containing a unique state (like $\Omega_\mathrm{one} =1$ in both figures). \label{f22}}

\end{figure}

For the case $d_2 =1$, the state space $\s_2$ is a segment (a 1-dimensional ball), hence the previous and next theorems are trivial. For $d_2 >1$, the previous theorem implies that $\s_2$ contains infinitely-many pure states. The next theorem recovers the (quantum-like) Bloch sphere with a yet unknown dimension $d_2$.

\begin{theorem}\label{bs}
	There is a set of fiducial measurements for which $\hat\s_2$ is a $d_2$-dimensional unit ball.
\end{theorem}

\begin{proof}
Lemma \ref{l.SGS-1} shows that there is an invertible real matrix $S$ such that for each $\hat G\in \hat\g_2$ the matrix $S \hat GS^{-1}$ is orthogonal. Let us redefine the set $\hat\s_2$ by transforming the states as $\hat\varphi \rightarrow \hat\varphi' = q S \hat\varphi$, where the number $q>0$ is chosen such that all pure states are unit vectors $|\hat\varphi'|^2 = \hat\varphi'^{\mbox{\tiny T}} \hat\varphi' =1$. This is possible because in the transformed state space, all pure states are related by orthogonal matrices $(SGS^{-1})$ which preserve the norm. Since Theorem~\ref{pure=boundary} also applies to the redefined set $\hat\s'_2$, it must be a unit ball. In what follows we define a new set of fiducial measurements $x'_i$ such that the Bloch representation (\ref{hat psi}) associated to the new fiducial probabilities $p(x'_i)$ coincides with the redefinition $\hat\varphi'$.

Requirement \ref{a.effects} tells that in $\hat\s'_2$, all tight effects are allowed measurements. For each unit vector $\hat\nu \in \mathbb{R}^{d_2}$ the function $\hat\Omega_{\hat\nu} (\hat\varphi')= (1+\hat\nu\t \hat\varphi')/2$ is a tight effect on the unit ball, and conversely, all tight effects on the unit ball are of this form. The new set of fiducial measurements $x'_i$ has effects $\hat\Omega_{x'_i} = \hat \Omega_{\hat\nu_i}$, where
\begin{equation}\label{thetai}
\hat\nu_1 = \left[ \begin{array}{c}
		1 \\ 0 \\ \vdots \\ 0
	\end{array} \right],\ 
\hat\nu_2 = \left[ \begin{array}{c}
		0 \\ 1 \\ \vdots \\ 0
	\end{array} \right],\ \ldots\ ,\ 
\hat\nu_{d_2} = \left[ \begin{array}{c}
		0 \\ 0 \\ \vdots \\ 1
	\end{array} \right] 
\end{equation}
is a fixed orthonormal basis for $\mathbb{R} ^{d_2}$. For any state $\hat\varphi'$ the new fiducial probabilities are $p(x'_i) = \Omega_{x'_i} (\hat\varphi')=$ \mbox{$(1+\hat \varphi'^i)/2$}, which implies $\hat\varphi'^i = 2[p(x'_i) -1/2]$. This is just (\ref{hat psi}) with the new fiducial measurements (note that $\hat\mu'_2  =\z$ and $\mu'^i_2 = \hat\Omega_{x'_i} (\z)=1/2$).
\end{proof}

In the rest of the paper, we will use the representation derived in Theorem~\ref{bs} above, where the generalized bit is represented by a unit ball. Moreover, we will drop the prime in $\hat\s'_2$, $x_i'$, $\hat\varphi'$ used in the proof, and simply write $\hat\s_2$, $x_i$, $\hat\varphi$.

As argued above, for each pure state $\varphi \in\s_2$ there is a binary measurement with associated effect
\begin{equation}\label{effect-state}
	\Omega_{\varphi} (\psi)= (1+\hat\varphi\t \hat\psi)/2\ ,
\end{equation}
such that $\hat \Omega_{\varphi} (\hat\varphi)=1$ and $\hat \Omega_{\varphi} (-\hat \varphi) =0$. In summary, there is a correspondence between tight effects and pure states in $\s_2$, and each pure state belongs to a distinguishable pair $\{\hat\varphi, -\hat\varphi\}$.

\subsection{Capacity and dimension}
\label{SubsecCD}

Requirements \ref{a.finiteness}, \ref{a.correlations} and \ref{a.subspaces} imply that a state space with finite capacity $c$ has finite dimension $d_c$, which generalizes Requirement~\ref{a.finiteness}. To see this, consider a system composed of $m$ generalized bits, with state space denoted by $\s_{2^{\times m}}$. Since $d_2$ is finite, equation (\ref{h dim}) implies that $\s_{2^{\times m}}$ has finite dimension. Due to the fact that perfectly distinguishable states are linearly independent, its capacity, denoted $c_m$, must be finite, too. Since systems with the same capacity are equivalent, we must have $c_m\neq c_n$ for $m\neq n$, and the sequence of integers $c_1, c_2, \ldots$ is unbounded. For any capacity $c$ there is a value of $m$ such that $c\leq c_m$, hence by Requirement \ref{a.subspaces} we have $\s_c\subset \s_{2^{\times m}}$, which implies that $\s_c$ is finite-dimensional.

In QT, the maximally-mixed state (\ref{maxmix}) has two convenient properties. First property: if $\mu_A$ and $\mu_B$ are the maximally-mixed states of systems $A$ and $B$, then the maximally-mixed state of the composite system $AB$ is 
\begin{equation}\label{mix prod}
	\mu_{AB} = \mu_A \otimes \mu_B\ .
\end{equation}
Second property: in the state space $\s_c$, there are $c$ pure distinguishable states $\psi_1, \ldots, \psi_c \in\s_c$ such that
\begin{equation}\label{mix decomp}
   \mu_c=\frac 1 c \sum_{a=1}^c \psi_a\ .
\end{equation}
Lemmas \ref{LemMaxMix} and \ref{LemMixture} show that these two properties hold for every theory satisfying our requirements. The following theorem exploits these properties to show that the capacity is multiplicative (one of the axioms in~\cite{Hardy5}).

\begin{theorem}
\label{LemCapacity}
If $c_A$ and $c_B$ are the capacities of systems $A$ and $B$, then the capacity of the composite system $AB$ is
\begin{equation}\label{claim1}
   c_{AB}=c_A c_B\ .
\end{equation}
\end{theorem}

\begin{proof}
Equation (\ref{mix decomp}) allows to write the maximally-mixed states of systems $A$ and $B$ as
\[
   \mu_A=\frac 1 {c_A}\sum_{a=1}^{c_A} \varphi_a^A\ ,
\quad \mu_B=\frac 1 {c_B}\sum_{b=1}^{c_B} \varphi^B_b\ ,
\]
where $\varphi^A_1,\ldots,\varphi^A_{c_A} \in \s_A$ are pure and distinguishable, and
$\varphi^B_1, \ldots, \varphi^B_{c_B} \in \s_B$ are pure and distinguishable, too. This and equation (\ref{mix prod}) imply
\begin{equation}\label{eqRepTP}
   \mu_{AB}=\mu_A\otimes \mu_B = \frac 1 {c_A c_B} \sum_{a=1}^{c_A} \sum_{b=1}^{c_B}
   \varphi^A_a\otimes \varphi^B_b\ .
\end{equation}
All states $\varphi^A_a \otimes \varphi^B_{b} \in \s_{AB}$ are distinguishable with the tensor-product measurement, therefore 
\begin{equation}\label{21}
	c_{AB}\geq c_A c_B\ . 
\end{equation}
Let $(\Omega_1, \ldots, \Omega_{c_{AB}})$ be a complete measurement on $AB$ which distinguishes the states $\psi_1, \ldots, \psi_{c_{AB}} \in \s_{AB}$; that is $\Omega_k (\psi_{k'}) =\delta_{k,k'}$. According to Lemma~\ref{l.pure} these states can be chosen to be pure. Since $\sum_{k=1}^{c_{AB}} \Omega_k(\mu_{AB}) =1$, there is at least one value of $k$, denoted $k_0$, such that 
\begin{equation}\label{i0}
	\Omega_{k_0} (\mu_{AB}) \leq 1/c_{AB}\ .
\end{equation}
The product of pure states $\varphi^A_1 \otimes \varphi^B_1$ is pure~\cite{tomo}, hence Requirement \ref{a.symmetry} tells that there is a reversible transformation $G\in \g_{AB}$ such that $G(\psi_{k_0}) = \varphi^A_1\otimes \varphi^B_1$. The measurement $(\Omega_1\circ G^{-1}, \ldots, \Omega_{c_{AB}} \circ G^{-1})$ distinguishes the states $G(\psi_1), \ldots, G(\psi_{c_{AB}})$. Inequality (\ref{i0}), the invariance of $\mu_{AB}$, expansion (\ref{eqRepTP}), the positivity of probabilities, and $(\Omega_{k_0} \circ G^{-1}) (\varphi^A_1 \otimes \varphi^B_1) =1$, imply
\begin{eqnarray*}
   \frac{1}{c_{AB}} &\geq& (\Omega_{k_0} \circ G^{-1}) (\mu_{AB}) 
	\\ &=& \frac{1}{c_A c_B} \sum_{a,b}\,  (\Omega_{k_0} \circ G^{-1}) (\varphi^A_a \otimes \varphi^B_b) \geq \frac 1 {c_A c_B}\ .
\end{eqnarray*}
This and (\ref{21}) imply (\ref{claim1}).
\end{proof}

It is shown in \cite{Hardy5} that the two multiplicativity formulas (\ref{h dim}) and (\ref{claim1}) imply the existence of a positive integer $r$ such that: for any $c$ the state space $\s_c$ has dimension
\begin{equation}\label{d-c-relation}
	d_c = c^r -1\ .
\end{equation}
The integer $r$ is a constant of the theory, with values $r=1$ for CPT and $r=2$ for QT.

\subsection{Recovering classical probability theory}

Let us consider all theories with $d_2=1$. In this case, equation (\ref{d-c-relation}) becomes $d_c =c-1$. In~\cite{Hardy5}, it is shown that the only GPT with this relation between capacity and dimension is CPT, as described in Subsection~\ref{s.instances}. We reproduce the proof for completeness.

\begin{theorem}
The only GPT with $d_2=1$ satisfying Requirements 1--5 is classical probability theory.
\end{theorem}

\begin{proof}
Let $\s_c$ be a state space and $(\Omega_1, \ldots, \Omega_c)$ a complete measurement which distinguishes the states $\psi_1, \ldots, \psi_c \in\s_c$. The vectors $\psi_1, \ldots, \psi_c \in \mathbb{R}^c$ are linearly independent; otherwise $\psi_{a} = \sum_{b\neq a} t_b\, \psi_b$ and $1= \Omega_a (\psi_a)$ $=$ $\sum_{b\neq a} t_b\, \Omega_a (\psi_b) =0$ gives a contradiction. Therefore, any state $\psi\in \s_c \subseteq \mathbb{R}^c$ can be written in this basis $\psi =\sum_a q_a\, \psi_a$ where $q_a = \Omega_a (\psi)$ turns out to be the probability of outcome $a$. The numbers $(q_1, \ldots, q_c)$ constitute a probability distribution, hence, there is a one-to-one correspondence between states in $\s_c$ and $c$-outcome probability distributions. This kind of set is called a $d_c$-simplex. A similar argument shows that the effects $\Omega_1, \ldots, \Omega_c$ are linearly independent. Hence, any effect $\Omega$ on $\s_c$ can be written as $\Omega = \sum_a h_a \Omega_a$, and the constraint $0\leq \Omega (\psi_a) \leq 1$ implies $0\leq h_a \leq 1$. In other words, every measurement on $\s_c$ is generated by the complete one.

Every reversible transformation on $\s_c$ is a symmetry of the $d_c$-simplex, that is, a permutation of pure states. Due to Requirement~\ref{a.symmetry}, there is a reversible transformation on the bit $\s_2$ which exchanges the two pure states. Using  Requirement~\ref{a.subspaces} inductively: if there is a transformation on $\s_{c-1}$ which exchanges two pure states and leaves the rest invariant, this transposition can be implemented on $\s_c$, also leaving all other pure states invariant. Therefore, all transpositions can be implemented in $\s_c$, and those generate the full group of permutations.
\end{proof}

\subsection{Reversible transformations for the generalized bit}\label{rev-trans-gbit}

In the rest of the paper, only theories with $d_2>1$ are considered. Theorem~\ref{bs} shows that $\hat\s_2$ is a $d_2$-dimensional unit ball. Equation (\ref{d-c-relation}) for $c=2$ implies that $d_2$ is odd. The pure states in $\hat \s_2$ are the unit vectors in $\mathbb{R}^{d_2}$. A reversible transformation $\hat G \in \hat \g_2$ maps pure states onto pure states, hence it preserves the norm and has to be an orthogonal matrix $\hat G\t = \hat G^{-1}$. Therefore $\hat \g_2$ is a subgroup of the orthogonal group O$(d_2)$.

Requirement \ref{a.symmetry} imposes that for any pair of unit vectors $\hat\varphi, \hat\varphi'$ there is $\hat G \in \hat \g_2$ such that $\hat G(\hat\varphi) = \hat\varphi'$. In other words, $\hat \g_2$ is transitive on the sphere~\cite{Onishchik1, Onishchik2}. According to Lemma~\ref{conncomp}, if $\hat \g_2$ is transitive on the sphere, then the largest connected subgroup $\hat {\cal C}_2 \subseteq \hat \g_2$ is also transitive on the  sphere. The matrix group $\hat {\cal C}_2$ is compact and connected, hence a Lie group (Theorem 7.31 in \cite{matrix_groups}). The classification of all connected compact Lie groups that are transitive on the sphere is done in \cite{Onishchik1, Onishchik2}. For odd $d_2$, the only possibility is $\hat{\cal C}_2 = \mathrm{SO} (d_2)$, except for $d_2 = 7$ where there are additional possibilities: $\hat {\cal C}_2= M \mathbb{G}_{2} M\t \subset {\rm SO}(7)$ for any $M \in \mathrm{O} (7)$, where $\mathbb{G}_{2}$ is the fundamental representation of the smallest exceptional Lie group \cite{group_book}. For even $d_2$, there are many more possibilities~\cite{Onishchik1, Onishchik2}, but equation~(\ref{d-c-relation}) implies that $d_2$ must be odd.

The stabilizer of the vector $\hat\nu_1$ defined in (\ref{thetai}) is the subgroup $\hat {\cal H}_2 = \{\hat G\in \hat\g_2: \hat G (\hat\nu_1) = \hat\nu_1 \}$. Each transformation $\hat H\in \hat {\cal H}_2$ has the form
\[
	\hat H= \left( \begin{array}{cc}
		1 & \z\t \\
		\z & \bar H \\
	\end{array} \right)\ ,
\]
where $\bar H \in \bar {\cal H}_2$ is the nontrivial part. If $\hat {\cal C}_2 = \mathrm{SO} (d_2)$ then $\mathrm{SO} (d_2-1) \subseteq \bar {\cal H}_2$. In the case $d_2=7$, if $\hat {\cal C}_2= M \mathbb{G}_{2} M\t$ then $\bar {\cal H}_2$ contains (up to the similarity $M$) the real 6-dimensional representation of $\mathrm{SU} (3)$ given by
\begin{equation}\label{su3}
	\bar H = \left( \begin{array}{cc}
		\mathrm{re}\, U & \mathrm{im}\, U \\
		-\mathrm{im}\, U & \mathrm{re}\, U \\
	\end{array} \right) \ ,
\end{equation}
where $\mathrm{re}\, U$ and $\mathrm{im}\, U$ are the real and imaginary parts of $U\in \mathrm{SU} (3)$ (see exercise 22.27 in \cite{group_book}). 


\subsection{Two generalized bits}\label{2gbits}

The joint state space of two $\s_2$ systems is denoted by $\s_{2,2}$. The multiplicativity of the capacity (\ref{claim1}) implies that $\s_{2,2}$ is equivalent to $\s_4$. However, we write $\s_{2,2}$ to emphasize the bipartite structure.

In what follows, instead of using the standard representation for bipartite systems (\ref{psiAB}) we generalize the Bloch representation to two generalized bits. A state $\psi_{AB} \in\s_{2,2}$ has Bloch representation $\hat\psi_{AB} = [\alpha, \beta, C]$ with
\begin{equation}\label{new-rep}
\begin{array}{lll}
	\alpha^i &=& 2p(x_i)-1
	\\ 
	\beta^j &=& 2p(y_j)-1
	\\ 
	C^{ij} &=& 4p(x_i, y_j) -2p(x_i) -2p(y_j) +1
\end{array}
\end{equation}
for $i,j=1, \ldots, d_2$. Note that $\alpha = \hat\psi_A$ and $\beta = \hat\psi_B$ are the reduced states in the Bloch representation (\ref{hat psi}). The correlation matrix can also be written as $C^{ij} = p(x_i, y_j) -p(x_i, \bar y_j) -p(\bar x_i, y_j) +p(\bar x_i, \bar y_j)$, and characterizes the correlations between subsystems. Product states have Bloch representation 
\begin{equation}\label{bloch-prod}
	(\varphi_A \otimes \varphi_B)^\wedge = \big[ \hat\varphi_A, \hat\varphi_B, \hat\varphi_A \hat\varphi_B\t \big]\ ,
\end{equation}
with rank-one correlation matrix. In QT, where $d_2=3$, two-qubit density matrices are often represented by $[\alpha, \beta, C]$ through formula (\ref{dm}). Definition (\ref{new-rep}) implies
\begin{equation}\label{bounent}
	-1\leq \alpha^i, \beta^j, C^{ij} \leq 1\ . 
\end{equation}

The invertible map ${\cal L} [\psi_{AB}] = \hat\psi_{AB}$ defined by (\ref{new-rep}) also determines the Bloch representation of effects $\hat\Omega =$ $\Omega \circ {\cal L}^{-1}$. In particular, the tensor-product of two effects of the form (\ref{effect-state}) is
\begin{equation}\label{prod-eff}
	(\Omega_{\varphi_{A}} \!\otimes \Omega_{\varphi_{B}}) ^\wedge [\alpha, \beta, C]  
=
	\big( 1+ \hat\varphi_A\t \alpha +\hat\varphi_B\t \beta + \hat\varphi_A\t C \hat\varphi_B \big)/4\ .
\end{equation}
The map ${\cal L}$ also determines the action of reversible transformation in the Bloch representation. Since ${\cal L}$ is affine but not linear, the action $\hat G= {\cal L}\circ G\circ {\cal L}^{-1}$ need not be linear. Identities (\ref{mix prod}, \ref{bloch-prod}) and $\hat\mu_2 = \z$ imply that the maximally-mixed state in $\hat\s_{2,2}$ is $\hat\mu_{2,2} = (\hat\mu_2, \hat\mu_2, \hat\mu_2 \hat\mu_2\t) = \z$. This and (\ref{inv state}) imply that transformations $\hat\g_{2,2}$ act on the generic vector $[\alpha, \beta, C]$ as matrices. In particular, local transformations $G_A, G_B \in \g_2$ act as
\begin{equation}\label{localt}
	(G_A \otimes G_B)^\wedge [\alpha, \beta, C] =  
[\hat G_A \alpha, \hat G_B \beta, \hat G_A C \hat G_B\t]
\ .
\end{equation}
Subsection \ref{rev-trans-gbit} concludes that $\hat \g_2$ consists of orthogonal matrices, and Lemma~\ref{sogs} shows that all transformations in $\hat\g_{2,2}$ are orthogonal, too. Orthogonal matrices preserve the norm of vectors, therefore all pure states $\psi \in \s_{2,2}$ satisfy
\begin{equation}\label{psnorm}
	|\hat\psi|^2 =
	|\alpha|^2 + |\beta|^2 + \mathrm{tr}(C\t C) = 3\ .
\end{equation}
The constant in the right-hand side can be obtained by letting $\hat\psi = [\alpha, \alpha, \alpha\alpha\t]$ with $|\alpha| = 1$.

\subsection{Consistency in the subspaces of two generalized bits}

In this subsection we use a trick introduced in \cite{Daki}: to impose the equivalence between a particular subspace of $\s_{2,2}$ and $\s_2$ (Requirement~\ref{a.subspaces}).

Consider the unit vector $\hat\nu_1$ from (\ref{thetai}) and the two distinguishable pure states $\hat\varphi_0 = \hat \nu_1$ and $\hat\varphi_1 = -\hat \nu_1$ from $\hat\s_2$. The four pure states $\varphi_{a,b} = \varphi_a \otimes \varphi_b \in \s_{2,2}$ can be distinguished with the complete measurement $\Omega_{a,b} = \Omega_{\varphi_a} \otimes \Omega_{\varphi_b}$ where $a,b \in \{0,1\}$. Formula (\ref{prod-eff}) implies 
\begin{eqnarray}
	\hat\Omega_{0,0} [\alpha, \beta, C] &=& (1+ \alpha^1 +\beta^1 +C^{1,1})/4\ ,
	\label{equ1} \\ \label{equ2}
	\hat\Omega_{1,1} [\alpha, \beta, C] &=& (1 -\alpha^1 -\beta^1 +C^{1,1})/4\ .
\end{eqnarray}
Requirement \ref{a.subspaces} implies that the subspace
\[
	\s_2' = \{\psi \in \s_{2,2} : (\Omega_{0,0} + \Omega_{1,1}) (\psi) =1\}\ ,
\]
is equivalent to $\s_2$. By adding (\ref{equ1}) plus (\ref{equ2}), it becomes clear that a state $\hat\psi=[\alpha, \beta, C]$ belongs to $\hat\s_2'$ if and only if $C^{1,1} =1$. Moreover, if $\psi\in\s'_2$, then it follows from $\hat\Omega_{0,1}(\hat\psi)\geq 0$ and $\hat\Omega_{1,0}(\hat\psi)\geq 0$ that $\alpha^1=\beta^1$.

\begin{theorem}
	The state space of a generalized bit has dimension three ($d_2=3$).
\end{theorem}

\begin{proof}
Recall that the case under consideration is odd $d_2$ larger than one. The space $\s_2' \subset \s_{2,2}$ is equivalent to $\s_2$, which is a $d_2$-dimensional unit ball. If $\varphi_{0,0}$ and $\varphi_{1,1}$ are considered the poles of this ball, then the equator is the set of states $\psi_\mathrm{eq}$ such that $\Omega_{0,0} (\psi_\mathrm{eq}) = \Omega_{1,1} (\psi_\mathrm{eq}) =1/2$. Equations \mbox{(\ref{equ1}, \ref{equ2})} tell that equator states have $\alpha^1 = \beta^1 =0$, and then 
\begin{equation}\label{cheq}
	\hat \psi_\mathrm{eq} = [ 
	\left( \begin{array}{c}
		0 \\ \bar\alpha 
	\end{array} \right),
	\left( \begin{array}{c}
		0 \\ \bar\beta 
	\end{array} \right),
	\left( \begin{array}{cc}
		1 & \bar\tau\t \\
		\bar\gamma & \bar C \\
	\end{array} \right)
]\ ,
\end{equation}
where $\bar\alpha, \bar\beta, \bar\gamma, \bar\tau \in \mathbb{R}^{d_2-1}$ and $\bar C \in \mathbb{R}^{(d_2-1) \times (d_2-1)}$. Consider the action of $G_A\otimes\id$ for $G_A\in\g_2$ on an equator state $\psi_\mathrm{eq}$.
Since $\hat\g_2$ is transitive on the unit sphere, if $\bar\gamma \neq \z$ then there is some $\hat G_A \in \hat\g_2$ such that the correlation matrix transforms into
\[
	\hat G_A \left( \begin{array}{cc}
		1 & \bar\tau\t \\
		\bar\gamma & \bar C \\
	\end{array} \right)
	=\left( \begin{array}{cc}
		\sqrt{1 + |\bar\gamma|^2} & ?\\
		\z &  ? \\
	\end{array} \right)\ ,
\]
which is in contradiction with (\ref{bounent}). Therefore $\bar\gamma = \z$, and by a similar argument $\bar\tau = \z$.

The stabilizer of $\hat\nu_1$ is the largest subgroup $\hat {\cal H}_2 \subset \hat\g_2$ which leaves $\hat\nu_1$ invariant (Subsection \ref{rev-trans-gbit}). For any pair $H_A, H_B \in {\cal H}_2$ the identity $\Omega_{a,b} \circ (H_A \otimes H_B) = \Omega_{a,b}$ holds, which implies that if $\psi_\mathrm{eq}$ belongs to the equator (\ref{cheq}) then
\[
	(H_A \otimes H_B) (\psi_\mathrm{eq})^\wedge
	= [ 
	\left( \begin{array}{c}
		0 \\ \bar H_A \bar\alpha
	\end{array} \right)\! ,
	\left( \begin{array}{c}
		0 \\ \bar H_B \bar\beta 
	\end{array} \right)\! ,
	\left( \begin{array}{cc}
		1 & \z\t  \\
		\z & \bar H_A \bar C \bar H_B\t \\
	\end{array} \right)]
\]
also belongs to the equator. The equator is a unit ball of dimension $d_2-1$. Since the set
\begin{equation}\label{sethh}
	\{ (H_A \otimes H_B) (\psi_\mathrm{eq}) : H_A, H_B \in {\cal H}_2 \}
\end{equation}
is a subset of the equator, the dimension of its affine span is at most $d_2-1$. 

Consider the case $\bar C =\z$. The normalization condition (\ref{psnorm}) implies $|\bar\alpha| =|\bar\beta| =1$. The set \mbox{$\{ \bar H_A \bar\alpha : \bar H_A \in \bar{\cal H}_2 \}$} has dimension $d_2-1$, and the same for \mbox{$\{ \bar H_B \bar\beta : \bar H_B \in \bar{\cal H}_2 \}$}. Therefore the set (\ref{sethh}) has dimension at least $2(d_2-1)$ generating a contradiction.

Consider the case $\bar C \neq \z$. The group action on $\bar C$ corresponds to the exterior tensor product \mbox{$\bar{\cal H}_2 \boxtimes \bar{\cal H}_2$} = \mbox{$\{\bar H_A \otimes \bar H_B:  \bar H_A, \bar H_B \in \bar{\cal H}_2\}$}. If \mbox{$d_2>3$} and \mbox{$\mathrm{SO}(d_2-1)$} $\subseteq \bar{\cal H}_2$ then $\bar{\cal H}_2$ is irreducible in $\mathbb{C} ^{d_2-1}$, and a simple character-based argument shows that \mbox{$\bar{\cal H}_2 \boxtimes \bar{\cal H}_2$} is irreducible in $(\mathbb{C} ^{d_2-1}) ^{\otimes 2}$ (see page~427 in \cite{group_book}). Hence the set
\begin{equation}\label{setC}
	\{ \bar H_A \bar C \bar H_B\t : \bar H_A, \bar H_B \in \bar {\cal H}_2 \} 
\end{equation}
has dimension \mbox{$(d_2 -1)^2$}, which conflicts with the dimensionality requirements of (\ref{sethh}). If $d_2=7$ and $\bar{\cal H}_2$ contains the representation of ${\rm SU}(3)$ given in (\ref{su3}), then the subgroup  
\[
	\bar H = \left( \begin{array}{cc}
		U & \z \\
		\z & U \\
	\end{array} \right)
\]
with $U \in \mathrm{SO} (3) \subset \mathrm{SU} (3)$ has two invariant $\mathbb{C}^{3}$ subspaces. Therefore the invariant subspaces of $\bar{\cal H}_2 \boxtimes \bar{\cal H}_2$ have at least dimension 9, and independently of $\bar C$, the set (\ref{setC}) has at least dimension 9, which conflicts with the dimensionality requirements of (\ref{sethh}). So the only possibility is $d_2=3$.
\end{proof}

From now on, only the case $d_2=3$ is considered. Subsection \ref{rev-trans-gbit} tells that $\mathrm{SO} (3) \subseteq \hat\g_2 \subseteq {\rm O}(3)$, which implies that either $\hat\g_2 = {\rm O}(3)$ and $\bar {\cal H}_2 = {\rm O}(2)$, or $\hat\g_2 = {\rm SO}(3)$ and $\bar {\cal H}_2 = {\rm SO}(2)$. 

Let us see that the first case is impossible. The group $\bar {\cal H}_2 = {\rm O}(2)$ is irreducible in $\mathbb{C}^{2}$, therefore $\bar{\cal H}_2 \boxtimes \bar{\cal H}_2$ is irreducible in $(\mathbb{C}^{2})^{\otimes 2}$. Three paragraphs above it is shown that $\bar C \neq \z$, hence the set (\ref{setC}) has dimension $(d_2-1)^2$, which is a lower bound for the one of (\ref{sethh}), which is larger than the allowed one ($d_2-1 =2$).

Let us address the second case. The group $\bar {\cal H}_2 = {\rm SO}(2)$ is irreducible in $\mathbb{R}^2$ but reducible in $\mathbb{C}^{2}$; so the previous argument does not hold. The vector space of $2\times 2$ real matrices decomposes into the subspace generated by rotations
\begin{equation}\label{rot}
	R_+ = \left( \begin{array}{cc}
		\cos v & \sin v \\
		-\sin v & \cos v \\
	\end{array} \right)\ ,
\end{equation}
and the one generated by reflections
\begin{equation}\label{ref}
	R_- = \left( \begin{array}{cc}
		\cos v & \sin v \\
		\sin v & -\cos v \\
	\end{array} \right)\ ,
\end{equation}
where $\det R_\pm =\pm 1$. For any pair $\bar H_A, \bar H_B \in \bar {\cal H}_2$ the matrix $\bar H_A R_+ \bar H_B\t$ is a rotation and the matrix $\bar H_A R_- \bar H_B\t$ is a reflection; therefore the 2-dimensional subspaces \mbox{(\ref{rot}) and (\ref{ref})} are invariant under $\bar{\cal H}_2 \boxtimes \bar{\cal H}_2$. Since the equator has dimension $d_2-1=2$, all matrices $\bar C \neq\z$ must be fully contained in one of the two subspaces spanned by (\ref{rot}) or (\ref{ref}), otherwise the dimension of the set (\ref{sethh}) would be too large again. For the same reason $\bar\alpha =\bar\beta =\z$. 

Depending on whether $\bar C$ is in the subspace generated by $R_+$ from~(\ref{rot}) or by $R_-$ from~(\ref{ref}), the states in the equator of $\s'_2$ are either $\hat\psi_\mathrm{eq}^+$ or $\hat\psi_\mathrm{eq}^-$, where
\[
	\hat\psi_\mathrm{eq}^\pm =[ 
	\left( \begin{array}{c}
		0 \\ 0 \\ 0 
	\end{array} \right),
	\left( \begin{array}{c}
		0 \\ 0 \\ 0
	\end{array} \right),
	\left( \begin{array}{ccc}
		1 & 0 & 0 \\
		0 & \cos v & \sin v \\
		0 & \mp\sin v & \pm\cos v
	\end{array} \right)
]\ .
\]
The proportionality constants in $\bar C \propto R_\pm$ are fixed by normalization (\ref{psnorm}). It turns out that both the symmetric case $\hat\psi_\mathrm{eq}^-$ and the antisymmetric case $\hat\psi_\mathrm{eq}^+$ correspond to different representations of the same physical theory ---that is, the corresponding state spaces (together with measurements and transformations) are \emph{equivalent} in the sense of Subsection~\ref{equivss}. To see this, define the linear map $\hat\tau: \hat\s_2 \to \hat\s_2$ as $\hat\tau(\alpha_1,\alpha_2,\alpha_3)\t:=(\alpha_1,\alpha_2,-\alpha_3)\t$; that is, a reflection in the Bloch ball. The equivalence transformation is defined as ${\cal L}:=\tau\otimes\id$ (in quantum information terms, this is a ``partial transposition''). This map respects the tensor product structure, leaves the set of product states invariant, and satisfies $\hat{\cal L}(\hat\psi_\mathrm{eq}^+)=\hat\psi_\mathrm{eq}^-$ \cite{Daki}.
In other words: we have reduced the discussion of the antisymmetric theory to that of the symmetric theory~\footnote{As a physical interpretation of the antisymmetric case, consider two observers who have never met before, but who have independently built devices to measure spin-$\frac 1 2$ particles in three orthogonal directions. If they never had the chance to agree on a common ``handedness'' of spatial coordinate systems, and happen to have chosen two different orientations, they will measure antisymmetric correlation matrices on shared quantum states. The ``three-bit nogo result'' from~\cite{Daki} can be interpreted as follows: if there is a third observer, then it is impossible that \emph{every} pair of parties measures antisymmetric correlation matrices.}, which will be considered for the rest of the paper.

The orthogonality of the matrices in $\hat\g_{2,2}$ implies that $\hat\s_2'$ is a 3-dimensional ball, and not just affinely related to it. Hence all states on the surface of the ball $\hat\s_2'  \subset \hat\s_{2,2}$ can be parametrized in polar coordinates $u\in[0,\pi)$ and $v\in[0,2\pi)$ as
\begin{eqnarray}\label{all-eq-st}
	\hat\psi (u,v) = \hspace{7cm} \\ \nonumber
	[\left( \begin{array}{c}
		\cos u \\ 0 \\ 0 
	\end{array} \right)\! ,
	\left( \begin{array}{c}
		\cos u \\ 0 \\ 0
	\end{array} \right)\! ,
	\left( \begin{array}{ccc}
		1 & 0 & 0 \\
		0 & \sin u \cos v & \sin u \sin v \\
		0 & \sin u \sin v & -\sin u \cos v
	\end{array} \right)]\ .
\end{eqnarray}
These states cannot be written as proper mixtures of other states from $\hat\s'_2$. It is easy to see that this implies that they are pure states in $\hat\s_{2,2}$.

\subsection{The Hermitian representation}
\label{SubsecHermitian}

In this subsection, a new (more familiar) representation is introduced, where states in $\s_2$ are represented by $2\times 2$ Hermitian matrices. For any state $\psi \in\s_2$ in the standard representation (\ref{psi state}), define the linear map
\begin{equation}\label{def L}
	{\cal L}[\psi] = \psi^0 \frac{\id- \sigma^1- \sigma^2- \sigma^3}{2} + \sum_{i=1}^{3} \psi^i \sigma^i \ .
\end{equation}
The Pauli matrices
\[
	\sigma^1 = 
	\left(\begin{array}{cc} 
		0 & 1	\\ 
		1 & 0
	\end{array}\right)\ ,\ \ 
	\sigma^2 = 
	\left(\begin{array}{cc} 
		0 & -\mathrm{i}	\\ 
		\mathrm{i} & 0
	\end{array}\right)\ ,\ \ 
	\sigma^3 = 
	\left(\begin{array}{cc} 
		1 & 0 \\ 
		0 & -1
	\end{array}\right)\ , 
\]
together with the identity $\id$ constitute an orthogonal basis for the real vector space of Hermitian matrices. In terms of the Bloch representation, the map (\ref{def L}) has the familiar form 
\[
	{\cal L} [\psi] = \frac 1 2 \left( \id + \sum_{i=1}^{3} \hat\psi^i \sigma^i \right)\ . 
\]
All positive unit-trace $2\times 2$ Hermitian matrices can be written in this way with $\hat\psi$ in the unit sphere. Since $\hat\s_2$ is a 3-dimensional unit sphere, the set ${\cal L} [{\cal S}_2]$ is the set of quantum states. The extreme points of ${\cal L} [{\cal S}_2]$ are the rank-one projectors: each pure state $\psi\in \s_2$ satisfies ${\cal L} [\psi] = |\psi \rangle\! \langle \psi|$, where the vector $|\psi\rangle \in \mathbb{C}^{2}$ is defined up to a global phase. Effect (\ref{effect-state}) associated to the pure state  $\varphi \in \s_2$ is
\begin{equation}\label{localmeas2}
	\Omega_{\varphi} (\psi) = 
\left( \Omega_\varphi \!\circ\! {\cal L}^{-1}\right)\left({\cal L}[\psi]\right) =
\mathrm{tr} \left( |\varphi \rangle\! \langle \varphi|\, {\cal L} [\psi] \right)\ .
\end{equation}
Note that the state $\varphi$ and its associated effect $\Omega_\varphi$ are both represented by $|\varphi \rangle\! \langle \varphi|$. The action of a reversible transformation $\hat G\in \hat\g_2 = \mathrm{SO} (3)$ in the Hermitian representation is 
\[
	{\cal L} [G (\psi)] = U {\cal L}[\psi] U^\dagger \, , 
\]
where $U \in \mathrm{SU} (2)$ is related to $\hat G$ via
\begin{equation}\label{rel U-O}
	\sum_{j=1}^3 \hat G^{ji} \sigma^j = 
	U \sigma^i U^\dagger\ , 
\end{equation}
and $\hat G^{ji}$ are the matrix components (equation VII.5.12 in~\cite{Haar-book}). In summary, the generalized bit in all theories satisfying $d_2>1$ and the requirements, is equivalent to the qubit in QT.

\subsection{Reconstructing quantum theory}

In this subsection, the main result of this work is proved. But before, let us introduce some notation. 

In QT, the state space with capacity $c$ and the corresponding group of reversible transformations are 
\begin{eqnarray*}
	{\cal S}_c^{\rm Q} &=& \{\rho \in \mathbb{C}^{c\times c}:\ \rho \geq 0,\ \mathrm{tr} \rho =1\}\ ,
	\\
	\g_c^{\rm Q} &=& \{U \otimes U^* : U\in \mathrm{SU} (c) \}\ .
\end{eqnarray*}
The joint state space of $m$ generalized bits is denoted by $\s_{2^{\times m}}$, and the corresponding group of reversible transformations by $\g_{2^{\times m}}$. The Hermitian representation of a state $\psi \in \s_{2^{\times m}}$ is defined to be ${\cal L}^{\otimes m} [\psi]$, where ${\cal L}^{\otimes m}:={\cal L}\otimes \cdots\otimes {\cal L}$, and ${\cal L}$ is defined in~(\ref{def L}). The map ${\cal L}^{\otimes m}$ acts independently on each tensor factor, hence it translates the tensor product structure from the standard representation (\ref{localmeas}, \ref{localtrans}, \ref{prod-state}) to the Hermitian one.  For example, if $\varphi\in \s_2$ is a pure state, then ${\cal L}^{\otimes m} [\varphi^{\otimes m}] = |\varphi \rangle\!\langle \varphi |^{\otimes m}$.
The notation 
\begin{eqnarray*}
	\s_{2^{\times m}} ^{\rm H} &=& {\cal L} ^{\otimes m} [\s_{2^{\times m}}]\ , \\
	\g_{2^{\times m}} ^{\rm H} &=& {\cal L} ^{\otimes m} \circ\g_{2^{\times m}} \circ ({\cal L} ^{\otimes m})^{-1}\ ,
\end{eqnarray*}
will be useful. The Hermitian representation of a state $\hat\psi_{AB}= [\alpha, \beta, C] \in \hat\s_{2,2}$ is
\begin{eqnarray}\label{dm}
	{\cal L}^{\otimes 2} [\psi_{AB}] = \hspace{6cm}
	\\ \nonumber
	\frac{1}{4} \Big[\id\! \otimes\! \id
+\sum_{i=1}^3 \alpha^i \sigma^i\! \otimes\! \id
+\sum_{j=1}^3 \beta^j \id\! \otimes\! \sigma^j
+\sum_{i,j=1}^3 C^{ij} \sigma^i\! \otimes\! \sigma^j \Big]\ .
\end{eqnarray}
The action of local transformations $G_A, G_B \in \g_2$ on $\psi_{AB} \in \s_{2,2}$ is
\begin{equation}\label{mat local}
	{\cal L}^{\otimes 2}\! \big[(G_A \otimes G_B) (\psi_{AB}) \big]  
	=
	(U_A \otimes U_B) \rho_{AB} (U_A \otimes U_B) ^\dagger  
\end{equation}
where $\rho_{AB} = {\cal L}^{\otimes 2} [\psi_{AB}]$ and $U_A, U_B \in \mathrm{SU} (2)$ are related to $G_A, G_B$ via (\ref{rel U-O}). Now, we are ready to prove

\begin{theorem}
The only GPT with $d_2>1$ satisfying Requirements 1--5 is quantum theory.
\end{theorem}

\begin{proof}
We start by reproducing an argument from~\cite{Daki} which shows that ${\cal S}_4^{\rm Q} \subseteq \s_{2,2} ^{\rm H}$. 
A particular family of pure states in $\s_{2,2}$ is $\psi (u) = \psi (u,0)$ defined in (\ref{all-eq-st}). The Hermitian representation of $\psi (u)$ is the projector ${\cal L} ^{\otimes 2} [\psi(u)] = |\psi (u) \rangle\!\langle \psi (u) |$ onto the \mbox{$\mathbb{C}^{2} \otimes \mathbb{C}^{2}$}-vector
\[
	| \psi (u) \rangle = \cos\frac{u}{2} |+\rangle \!\otimes\! |+\rangle+\sin\frac u 2 |-\rangle \!\otimes\! |-\rangle\ ,
\]
where $|+\rangle= (1,1)\t/\sqrt{2}$ and $|-\rangle= (-1,1)\t/\sqrt{2}$. From the Schmidt decomposition, it follows that all rank-one projectors in $\mathbb{C}^{4\times 4}$ can be written as \mbox{$(U_A \otimes U_B) |\psi (u) \rangle\!\langle \psi (u) |(U_A \otimes U_B)^\dagger$} for some value of $u$ and some local unitaries $U_A , U_B \in \mathrm{SU}(2)$. Thus, all rank-one projectors are pure states in $\s_{2,2}^{\rm H}$. Their mixtures generate all of $\s_4 ^{\rm Q}$, therefore $ \s_4 ^{\rm Q} \subseteq \s_{2,2}^{\rm H}$. 

Direct calculation shows that
\begin{equation}\label{L isometry}
   {\rm tr}\!\left( {\cal L}^{\otimes 2}[\psi]\, {\cal L}^{\otimes 2} [\psi']\right)
=\frac 1 4 + \frac 1 4 \left[ \alpha\t \alpha' +\beta\t \beta' +{\rm tr}(C\t C') \right]
\end{equation}
for any pair of states $\hat\psi = [\alpha, \beta, C]$ and $\hat\psi' = [\alpha', \beta', C']$ from $\hat \s_{2,2}$. Lemma~\ref{sogs} shows that all $\hat G\in\hat\g_{2,2}$ are orthogonal matrices. Therefore, the Euclidean inner product between states, as in the right-hand side of (\ref{L isometry}), is preserved by the action of any $\hat G \in \hat\g_{2,2}$. Equality (\ref{L isometry}) maps this property to the Hermitian representation: any  $H \in \g_{2,2} ^{\rm H}$ preserves the Hilbert-Schmidt inner product between states:
\begin{equation}
   {\rm tr}\! \left[H(\rho)\, H(\rho')\right] = {\rm tr}(\rho\, \rho')\ ,
   \label{eqIsometry}
\end{equation}
for all $\rho, \rho' \in \s_{2,2} ^{\rm H}$. 

For any pure state $\varphi \in \s_2$, the rank-one projector $|\varphi \rangle\!\langle \varphi| \otimes |\varphi \rangle\!\langle \varphi| = {\cal L}^{\otimes 2}[\varphi \otimes \varphi]$ is a pure state in $\s_{2,2} ^{\rm H}$, and ${\rm tr} (|\varphi \rangle\!\langle \varphi| \otimes |\varphi \rangle\!\langle \varphi|\, \rho) = (\Omega_\varphi \otimes \Omega_\varphi) \circ ({\cal L}^{\otimes 2})^{-1} (\rho)$ is a measurement on $\s_{2,2} ^{\rm H}$. Any rank-one projector $|\psi \rangle\!\langle \psi| \in \mathbb{C}^{4 \times 4}$ is a pure state in $\s_{2,2} ^{\rm H}$, hence there is $H\in \g_{2,2} ^{\rm H}$ such that \mbox{$H(|\varphi \rangle\!\langle \varphi| \otimes |\varphi \rangle\!\langle \varphi|)$} $= |\psi \rangle\!\langle \psi|$. Composing the transformation $H$ with the effect $\Omega_\varphi \otimes \Omega_\varphi$ generates the effect $(\Omega_\varphi\otimes\Omega_\varphi)\circ({\cal L}^{\otimes 2})^{-1}\circ H^{-1}$, which maps any $\rho\in\s_{2,2}^{\rm H}$ to
\begin{equation}\label{psi_eff}
	{\rm tr}\! \left[ |\varphi \rangle\!\langle \varphi| \otimes |\varphi \rangle\!\langle \varphi|\, H^{-1}(\rho) \right] = 
	{\rm tr}\! \left[ |\psi \rangle\!\langle \psi|\, \rho \right]\ ,
\end{equation}
where (\ref{eqIsometry}) has been used. In summary, every rank-one projector $|\psi \rangle\!\langle \psi| \in \mathbb{C}^{4 \times 4}$ has an associated effect (\ref{psi_eff}) which is an allowed measurement on $\s_{2,2}^{\rm H}$, and these generate all quantum effects.

We have seen that all quantum states $\s_4^{\rm Q}$ are contained in $\s_4 ^{\rm H}$, but can there be other states? If so, the associated Hermitian matrices should have a negative eigenvalue (note that all states in the Hermitian representation (\ref{dm}) have unit trace). If $\rho$ has a negative eigenvalue and $|\psi\rangle$ is the corresponding eigenvector, then the associated measurement outcome (\ref{psi_eff}) has negative probability. Hence, we conclude that $\s_4^{\rm Q} = \s_{4} ^{\rm H}$, and similarly for the measurements.

All reversible transformations $H \in \g_4 ^{\rm H}$ map pure states to pure states, that is, rank-one projectors to rank-one projectors. According to Wigner's Theorem~\cite{Bargmann}, every map of this kind can be written as $H(|\psi\rangle\! \langle\psi|)= (U|\psi\rangle) (U|\psi\rangle)^\dagger$, where $U$ is either unitary or anti-unitary. If $U$ is anti-unitary, it follows from Wigner's normal form~\cite{Wigner} that there is a two-dimensional $U$-invariant subspace spanned by two orthonormal vectors $|\theta_0 \rangle, |\theta_1 \rangle \in \mathbb{C}^4$ such that $U\left( t_0 |\theta_0\rangle + t_1 |\theta_1\rangle \right)$ equals either $\bar t_0 |\theta_0\rangle+ \bar t_1 |\theta_1\rangle$ or $\bar t_1 e^{is} |\theta_0\rangle +\bar t_0 e^{-is} |\theta_1\rangle$ for some $s\in \mathbb{R}$. In both cases, $U$ acts as a reflection in the corresponding Bloch ball, which contradicts Requirement~\ref{a.subspaces} because we know that $\g_2=SO(3)$. Therefore $\g_4 ^{\rm H} \subseteq \g_4^{\rm Q}$.

We know that $\g_{2,2} ^{\rm H}$ contains all local unitaries. Since this group is transitive on the pure states, it contains at least one unitary which maps a product state to an entangled state. It is well-known~\cite{Bremner} that this implies that the corresponding group of unitaries constitutes a universal gate set for quantum computation; that is, it generates every unitary operation on $2$ qubits. This proves that $\g_4 ^{\rm H}= \g_4^{\rm Q}$.

Consider $m$ generalized bits as a composite system. From the previously discussed case of $\s_{2,2}$, we know that every unitary operation on every pair of generalized bits is an allowed transformation on $\s_{2^{\times m}} ^{\rm H}$. But two-qubit unitaries generate all unitary transformations~\cite{Bremner}, hence $\g_{2^m}^{\rm Q} \subseteq \g_{2^{\times m}} ^{\rm H}$. By applying all these unitaries to $|\varphi\rangle\! \langle\varphi|^{\otimes m}$, all pure quantum states are generated, hence $\s_{2^m}^{\rm Q} \subseteq \s_{2^{\times m}} ^{\rm H}$. Reasoning as in the $\s_{2,2}$ case, for every rank-one projector $|\psi\rangle\! \langle\psi|$ acting on $\mathbb{C}^{2^m}$, the associated effect which maps $\rho\in\s_{2^{\times m}}^{\rm H}$ to ${\rm tr}(|\psi\rangle\! \langle\psi| \rho)$ is an allowed measurement outcome on $\s_{2^{\times m}}^{\rm H}$. This implies that all matrices in $\s_{2^{\times m}} ^{\rm H}$ have positive eigenvalues, therefore $\s_{2^{\times m}} ^{\rm H} = \s_{2^m} ^{\rm Q}$ and $\g_{2^{\times m}} ^{\rm H} = \g_{2^ m} ^{\rm Q}$.

The remaining cases of capacities $c$ that are not powers of two are treated by applying Requirement \ref{a.subspaces}, using that $\s_c\subset\s_{2^m}$ for large enough  $m$.
\end{proof}

\section{Conclusion}\label{s.conclusions}

We have imposed five physical requirements on the framework of generalized probabilistic theories. These requirements are simple and have a clear physical meaning in terms of basic operational procedures. It is shown that the only theories compatible with them are CPT and QT. If Requirement \ref{a.symmetry} is strengthened by imposing the continuity of reversible transformations, then the only theory that survives is QT. Any other theory violates at least one of the requirements, hence the relaxation of each one constitutes a different way to go beyond QT.

The standard formulation of QT includes two postulates which do not follow from our requirements: (i) the update rule for the state after a measurement, and (ii) the Schr\"odinger equation. If desired, these can be incorporated in our derivation of QT by imposing the following two extra requirements: (i) {\em if a system is measured twice ``in rapid succession" with the same measurement, the same outcome is obtained both times} \cite{Hardy5}, and (ii) {\em closed systems evolve reversibly and continuously in time}. 

This derivation of QT contains two steps which deserve a special mention. First, a direct consequence of Requirement \ref{a.subspaces} is that $\s_2$ is fully surrounded by pure states, which together with Requirement \ref{a.symmetry} implies that $\s_2$ is a ball. Second, this ball has dimension three, since $d=3$ is the only value for which SO($d-1$) is reducible in $\mathbb{C}^d$.

Modifications and generalizations of QT are of interest in themselves, and could be essential in order to construct a QT of gravity. Some well-known attempts~\cite{Gleason, Weinberg} have shown that straightforward modifications of QT's mathematical  formalism quickly lead to inconsistencies, such as superluminal  signaling~\cite{Gisin}. This work provides an alternative way to proceed. We have shown that the Hilbert space formalism of QT follows from five simple physical requirements. This gives five different consistent ways to go beyond QT, each obtained by relaxing one of our requirements.


\acknowledgments

The authors are grateful to Anne Beyreuther, Jens Eisert, Volkher Scholz, Tony Short, and Christopher Witte for discussions. Special thanks to Lucien Hardy for pointing out the multiplicity of groups that are transitive on the sphere and, correspondingly, the need to address the 7-dimensional ball as a special case. Llu\'{\i}s Masanes is financially supported by Caixa Manresa, and benefits from the Spanish MEC project TOQATA (FIS2008-00784) and QOIT (Consolider Ingenio 2010), EU Integrated Project SCALA and STREP project NAMEQUAM.
Markus M\"uller was supported by the EU (QESSENCE).
Research at Perimeter Institute is supported by the Government of Canada through Industry Canada and by the Province of Ontario through the Ministry of Research and Innovation.

\appendix

\section{Lemmas}

\begin{lemma}\label{LemMaxMixUnique}
In any state space $\s_c$, the only state $\psi\in\s_c$ which is invariant under all reversible transformations 
\begin{equation}\label{l1e1}
	G(\psi)= \psi\ \mbox{ for all }\ G\in \g_c\ , 
\end{equation}
is the maximally-mixed state $\mu_c$, defined in (\ref{maxmix}).
\end{lemma}

\begin{proof}
Suppose $\psi\in\s_c$ satisfies (\ref{l1e1}). Any state can be written as a mixture of pure states: $\psi=\sum_k q_k\, \psi_k$. Normalization $\int_{\mathcal{G}_c} \!\! dG= 1$, condition (\ref{l1e1}), the linearity of $G$, the purity of all $\psi_k$, the definition of $\mu_c$, and $\sum_k q_k=1$, imply
\begin{eqnarray*}
   \psi&=& \int_{\mathcal{G}_c} \!\!\psi\, dG  =  \int_{\mathcal{G}_c} \!\!G(\psi)\, dG
	\\ &=& \sum_k q_k  \int_{\mathcal{G}_c} \!\! G(\psi_k)\, dG
	= \sum_k q_k\, \mu_c = \mu_c\ ,
\end{eqnarray*}
which proves the claim.
\end{proof}

\begin{lemma}\label{l.SGS-1}
	If $\g$ is a compact real matrix group, then there is a real matrix $S>0$ such that for each $G\in \g$ the matrix $SGS^{-1}$ is orthogonal.
\end{lemma}

\begin{proof}
Since the group $\g$ is compact, there is an invariant Haar measure \cite{Haar-book}, which allows us to define
\[
	P= \int_\g \!\! G\t G\, dG\ .
\]
Since each $G$ is invertible, the matrix $G\t G$ is strictly positive, and $P$ too. Define $S=\sqrt P >0$ where both $S,S^{-1}$ are real and symmetric. For any $G\in \g$ we have $(SGS^{-1})\t (SGS^{-1}) =\id$, which implies orthogonality.
\end{proof}

\begin{lemma}\label{LemMaxMix}
If $\mu_A$ and $\mu_B$ are the maximally-mixed states of the state spaces $\s_A$ and  $\s_B$, then the maximally-mixed state of the composite system $\s_{AB}$ is 
\[
	\mu_{AB} = \mu_A \otimes \mu_B\ .
\]
\end{lemma}

\begin{proof}
The pure states $\psi^A$ in $\s_A$ linearly span $\mathbb{R}^{d_A+1}$, and the pure states $\psi^B$ in $\s_B$ linearly span $\mathbb{R}^{d_B+1}$. Therefore, pure product states $\psi^A \otimes \psi^B$ span $\mathbb{R}^{d_A+1} \otimes \mathbb{R}^{d_B+1}$. In particular, the maximally-mixed state (\ref{maxmix}) of $\s_{AB}$ can be written as
\begin{equation}\label{la1}
	\mu_{AB} 
	= \sum_{a,b} t_{a,b}\, \psi_a^A \otimes \psi^B_b\ , 
\end{equation}
where $t_{a,b}\in\mathbb{R}$ are not necessarily positive coefficients,
and all $\psi^A_a, \psi^B_b$ are pure. From definition (\ref{psi state}), the first component of the vector equality (\ref{la1}) implies $\sum_{a,b} t_{a,b}=1$. The maximally-mixed state is invariant under all reversible transformations, in particular the local ones
\begin{eqnarray*}
   \mu_{AB} &=& \int_{\mathcal{G}_A} \hspace{-3mm} dG_A \int_{\mathcal{G}_B} \hspace{-3mm} dG_B\, (G_A \otimes G_B) (\mu_{AB}) \\
   &=& \sum_{a,b} t_{a,b} \left[\int_{\mathcal{G}_A} \hspace{-3mm} dG_A\, G_A (\psi_a^A)\right] 
	\!\otimes\!
	\left[ \int_{\mathcal{G}_B} \hspace{-3mm}dG_B\, G_B (\psi^B_b)\right]\\
   &=&\sum_{a,b} t_{a,b}\, \mu_A\otimes \mu_B = \mu_A\otimes \mu_B\ ,
\end{eqnarray*}
where the same tricks from Lemma~\ref{LemMaxMixUnique} have been used.
\end{proof}

\begin{lemma}\label{l.pure}
	For every tight effect $\Omega$, there is a pure state $\psi$ such that $\Omega(\psi) =1$. Also, if a measurement $\Omega_1,\ldots,\Omega_n$ distinguishes $n$ states $\psi_1,\ldots,\psi_n$, then these states can be chosen pure.
\end{lemma}

\begin{proof}
By definition, for each tight effect $\Omega$ there is a (not necessarily pure) state $\psi'$ such that $\Omega(\psi') =1$. Every $\psi'$ can be written as a mixture of pure states $\psi_k$, that is $\psi' = \sum_k q_k\, \psi_k$ with $q_k >0$ and $\sum_k q_k =1$. Effects are linear functions such that $\Omega(\psi) \leq 1$ for any state $\psi$. Therefore, it must happen that all pure states $\psi_k$ in the above decomposition satisfy $\Omega(\psi_k) =1$. 

To prove the second part, let $\psi'_1, \ldots, \psi'_n$ be the states that are distinguished by the measurement, that is $\Omega_a(\psi'_b) = \delta_{a,b}$. Every $\psi'_b$ can be written as a convex combination of pure states $\psi'_b=\sum_k q_k\, \psi_{b,k}$. But effects are linear functions such that $0\leq \Omega(\psi) \leq 1$ for any state $\psi$. Hence $\Omega(\psi'_b)=0$ is only possible if $\Omega(\psi_{b,k})=0$ for all $k$, and similarly for the case $\Omega(\psi'_b)=1$. It follows that $\Omega_a (\psi_{b,1}) = \delta_{a,b}$.
\end{proof}

\begin{lemma}
\label{LemMixture}
If $\mathcal{S}_c$ is a state space with capacity $c\geq 1$ and $\mu_c$ the corresponding maximally-mixed state, then there are $c$ pure distinguishable states $\psi_1,\ldots, \psi_c \in \s_c$ such that
\[
   \mu_c=\frac 1 c \sum_{a=1}^c \psi_a\ .
\]
\end{lemma}

\proof
Since $\mathcal{S}_1$ contains a single state the claim is trivially true for $c=1$. Since $\mathcal{S}_2$ is the $d_2$-dimensional unit ball, two antipodal points $\hat \varphi_1$ and $\hat \varphi_2 =-\hat \varphi_1$ are pure, distinguishable and satisfy 
\begin{equation}\label{B4}
	\mu_2=\frac 1 2(\varphi_1+ \varphi_2)\ .
\end{equation}
Now, consider the joint state space of $n$ generalized bits, denoted $\mathcal{S}_{2^{\times n}}$. Lemma~\ref{LemMaxMix} and (\ref{B4}) imply that the maximally-mixed state of $\mathcal{S}_{2^{\times n}}$ is
\begin{equation}\label{sum rep}
   \mu_{(n)}= (\mu_2)^{\otimes n}= \frac 1 {2^n} \sum_{a_i \in \{1,2\}} \varphi_{a_1}\otimes \cdots \otimes \varphi_{a_n}\ .
\end{equation}
The states $\varphi_{a_1} \otimes  \cdots \otimes \varphi_{a_n} \in \mathcal{S}_{2^{\times n}}$ for all  $a_1, \ldots, a_n \in\{1,2\}$ are perfectly distinguishable by the corresponding product measurement, hence the capacity of $\mathcal{S}_{2^{\times n}}$, denoted $c_n$, satisfies 
\begin{equation}\label{B7}
	c_n \geq 2^n\ . 
\end{equation}
Let $(\Omega_1,\ldots, \Omega_{c_n})$ be a complete measurement which distinguishes the states $\psi_1, \ldots, \psi_{c_n} \in \s_{2^{\times n}}$. According to Lemma~\ref{l.pure} these states can be chosen to be pure. Since $\sum_{k=1}^{c_n} \Omega_k (\mu_{(n)}) =1$, there is at least one value of $k$, denoted $k_0$, such that 
\begin{equation}\label{B8}
	\Omega_{k_0} (\mu_{(n)})\leq 1/c_n\ . 
\end{equation}
The state $\varphi_1 \in \s_2$ from (\ref{B4}) is pure, hence $\varphi_1^{\otimes n} \in \s_{2^{\times n}}$ is pure too. Requirement \ref{a.symmetry} tells that there is a reversible transformation $G$ acting on $\s_{2^{\times n}}$ such that $G(\psi_{k_0}) = \varphi_1 ^{\otimes n}$. The measurement $(\Omega_1 \circ G^{-1}, \ldots, \Omega_{c_n} \circ G^{-1})$ distinguishes the states $G(\psi_1), \ldots, G(\psi_{c_n})$. Inequality (\ref{B8}), the invariance of $\mu_{(n)}$ under $G$, expansion (\ref{sum rep}), the positivity of probabilities, and $(\Omega_{k_0} \circ G^{-1}) (\varphi_1^{\otimes n}) =1$, imply
\begin{eqnarray*}
   \frac{1}{c_n} &\geq& (\Omega_{k_0} \circ G^{-1}) (\mu_{(n)}) 
	\\ &=& \frac{1}{2^n}\!\! \sum_{a_i \in \{1,2\}} (\Omega_{k_0} \circ G^{-1}) (\varphi_{a_1} \otimes \cdots \otimes \varphi_{a_n}) \geq \frac 1 {2^n}\ .
\end{eqnarray*}
This and (\ref{B7}) imply $c_n = 2^n$. This together with (\ref{sum rep}) shows the assertion of the lemma for state spaces whose capacity is a power of two. The rest of cases are shown by induction.

Let us prove that if the claim of the lemma holds for a state space with capacity $c$, with $c>1$, then it holds for a state space with capacity $c-1$ too. The induction hypothesis tells that there is a complete measurement $(\Omega_1, \ldots, \Omega_c)$ which distinguishes the pure states $\psi_1, \ldots, \psi_c \in \mathcal{S}_c$, and $\mu_c = \frac 1 c \sum_{k=1}^c \psi_k \in \s_c$ is the corresponding maximally-mixed state. Requirement \ref{a.subspaces} tells that the state space $\s_{c-1}$ is equivalent to
\[
	\s_{c-1}' =\{\psi\in \s_c\, |\, \Omega_1 (\psi) +\cdots +\Omega_{c-1} (\psi) =1 \}\ .
\]
According to Requirement \ref{a.subspaces}, for each $G\in \g_{c-1}$ there is $G'\in \g_c$ which implements $G$ on ${\cal S}_{c-1}'$. Hence $G'(\psi_k) \in  {\cal S}_{c-1}'$ for $k=1,\ldots, c-1$, which implies $(\Omega_c\circ G')(\psi_k)=0$ for those $k$, and
\begin{eqnarray*}
   \frac 1 c &=& \Omega_c(\mu_c) = (\Omega_c \circ G') (\mu_c)
	\\ &=& \frac 1 c \sum_{k=1}^c\ (\Omega_c \circ G')  (\psi_k)
	=\frac 1 c (\Omega_c \circ G') (\psi_c)\ ,
\end{eqnarray*}
and $(\Omega_c\circ G') (\psi_c)=1$. Requirement \ref{a.subspaces} tells that the set
\[
    {\cal S}_1'= \{\psi \in {\cal S}_c \,|\, \Omega_c(\psi)=1\}
\]
is equivalent to $\mathcal{S}_1$, which contains a single state. This and $\Omega_c(\psi_c)=1$ imply that $G'(\psi_c) =\psi_c$ and then $G'(\mu_{c-1}') =\mu_{c-1}'$, where we define
\[
   \mu_{c-1}' = \frac{1}{c-1} \sum_{k=1}^{c-1} \psi_k = \frac{c}{c-1} \mu_c - \frac{1}{c-1} \psi_c \ \in {\cal S}_{c-1}'\ .
\]
For any $G\in \g_{c-1}$ the corresponding $G'$ satisfies $G'(\mu_{c-1}')= \mu_{c-1}'$. Due to Lemma~\ref{LemMaxMixUnique} the invariant state $\mu_{c-1}'$ must be the maximally-mixed state in $\s_{c-1}'$, which has the claimed form, and  Requirement \ref{a.subspaces} extends this to $\s_{c-1}$. \qed

\begin{lemma}\label{conncomp}
   Let $\s$ be a state space such that the subset of pure states ${\cal P}$ is a connected topological manifold. If the corresponding group of transformations $\g$ is transitive on ${\cal P}$, then the largest connected subgroup ${\cal C \subseteq G}$ is transitive on ${\cal P}$, too.
\end{lemma}

\begin{proof}
This proof involves basic notions of point set topology.

Since $\g$ is compact, it is the union of a finite number of (disjoint) connected components $\g= {\cal C}_1 \cup \cdots \cup {\cal C}_n$. If $n=1$ the lemma is trivial. Let ${\cal C}$ be the connected component ${\cal C}_i$ containing the identity matrix $\id$, which is the largest connected subgroup of $\g$. Each connected component ${\cal C}_i$ is clopen (open and closed), compact and a coset of the group: ${\cal C}_i = G_i\circ {\cal C}$ for some $G_i \in \g$ \cite{group_book}.

Pick $\psi\in {\cal P}$, and consider the continuous surjective map $f: \g \rightarrow {\cal P}$, defined by $f(G)= G(\psi) \in{\cal P}$. Since ${\cal C}$ is compact $f({\cal C}) \subseteq {\cal P}$ is compact too. Since the manifold ${\cal P}$ is in particular a Hausdorff space, $f({\cal C})$ is closed. Consider the set ${\cal D} =f^{-1} (f({\cal C}))$. If two group elements $G,H\in\g$ are in the same component, that is $G^{-1}H\in{\cal C}$, then $G\in{\cal D}$ implies $H\in{\cal D}$, using that ${\cal C}$ is a normal subgroup of $\g$. This implies that ${\cal D}$ is the union of some connected components ${\cal C}_i$, and so is $\g\backslash {\cal D}$. In particular, $\g\backslash {\cal D}$ is compact, thus $f(\g\backslash {\cal D})$ is compact, hence closed. Therefore, $f({\cal C})={\cal P}\backslash f(\g\backslash {\cal D})$ is open. We have thus proven that $f({\cal C})\neq \emptyset$ is clopen.
Since ${\cal P}$ is connected, it follows that $f({\cal C})={\cal P}$.
\end{proof}

The following lemma shows that there are transformations for two generalized bits which perform the ``classical" swap and the ``classical" controlled-not in a particular basis. Note that these transformations do not necessarily swap other states that are not in the given basis, as in QT. However, they implement a minimal amount of reversible computational power which exceeds, for example, that of boxworld, where no controlled-not operation is possible~\cite{boxworld}.

\begin{lemma}\label{PseudoSwap}
For each pair of distinguishable states $\varphi_0, \varphi_1\in \s_2$, there are transformations
$G_\mathrm{swap}, G_\mathrm{cnot} \in \g_{2,2}$ such that
\begin{eqnarray}\label{pseudoSwap}
	G_\mathrm{swap} (\varphi_a \otimes \varphi_b) &=& \varphi_b \otimes \varphi_a\ ,
	\\ \label{pseudoCNOT}
	G_\mathrm{cnot} (\varphi_a \otimes \varphi_b) &=& \varphi_a \otimes \varphi_{a \oplus b}\ ,
\end{eqnarray}
for all $a,b \in \{0,1\}$, where $\oplus$ is addition modulo 2.
\end{lemma}

\begin{proof}
Let $(\Omega_0, \Omega_1)$ be the measurement which distinguishes $(\varphi_0, \varphi_1)$, that is $\Omega_a (\varphi_b)= \delta_{a,b}$. Define $\psi_{a,b} = \varphi_a \otimes \varphi_b$ and $\Omega_{a,b} = \Omega_a \otimes \Omega_b$ for $a,b \in \{0,1\}$. Define
\begin{eqnarray*}
   \s'_3 &=& \{\psi\in\s_{2,2}\,\,|\,\, (\Omega_{0,1} +\Omega_{1,0} +\Omega_{1,1}) (\psi)=1\},\\
   \s'_2 &=& \{\psi\in\s_{2,2}\,\,|\,\, (\Omega_{0,1}+\Omega_{1,0})(\psi)=1\}\ ,
\end{eqnarray*}
and note that $\s'_2 \subset \s'_3 \subset \s_{2,2}$. At the end of Subsection~\ref{gener-bit}, it is shown that two distinguishable states $\varphi_0, \varphi_1\in\s_2$ have Bloch representation satisfying $\hat\varphi_0=-\hat\varphi_1$. According to Requirement~\ref{a.symmetry} there is $G\in\g_2$ such that $G(\varphi_0) =\varphi_1$, and by linearity, $\hat G(\hat\varphi_1)=-\hat G(\hat\varphi_0)=-\hat\varphi_1=\hat\varphi_0$.
Requirement \ref{a.subspaces} implies that there is a transformation $G'_\mathrm{swap}$ for $\s_3'$ such that $G'_\mathrm{swap} (\psi_{0,1}) = \psi_{1,0}$ and $G'_\mathrm{swap} (\psi_{1,0}) = \psi_{0,1}$. According to Lemma~\ref{LemMixture}, the maximally-mixed state in $\s'_3$ can be written as $\mu_3' = (\psi_{0,1} +\psi_{1,0} +\psi_{1,1})/3$. Equalities $G'_\mathrm{swap} (\mu_3') = \mu_3'$ and $G'_\mathrm{swap} (\psi_{0,1}+ \psi_{1,0}) = \psi_{0,1}+ \psi_{1,0}$ imply that $G'_\mathrm{swap} (\psi_{1,1}) = \psi_{1,1}$. Using Requirement \ref{a.subspaces} again, there is a reversible transformation $G_\mathrm{swap} \in \g_{2,2}$ which implements $G'_\mathrm{swap}$ in the subspace $\s'_3$. Repeating the argument with the maximally-mixed state (now in $\s_{2,2}$) we conclude that  $G_\mathrm{swap} (\psi_{1,1}) = \psi_{1,1}$, hence $G_\mathrm{swap}$ satisfies (\ref{pseudoSwap}).

The existence of $G_\mathrm{cnot}$ is shown similarly, by exchanging the roles of $\psi_{0,1}$ and $\psi_{1,1}$.
\end{proof}

\begin{lemma}\label{sogs}
Reversible transformations for two generalized bits in the Bloch representation (\ref{new-rep}) are orthogonal:
\[
	\hat \g_{2,2} \subseteq \mathrm{O} (d_4)\ .
\] 
\end{lemma}

\begin{proof}
In Subsection \ref{2gbits} the Bloch representation for two generalized bits is defined, and it is argued that reversible transformations $\hat \g_{2,2}$ act on $[\alpha ,\beta ,C]$ as matrices. In particular, local transformations (\ref{localt}) are 
\begin{equation}\label{3-action}
	(G_A \otimes G_B)^\wedge =
	\left(\begin{array}{ccc}
		\hat G_A &\z &\z \\
		\z & \hat G_B &\z \\	
		\z &\z & \hat G_A \otimes \hat G_B
	\end{array}\right) ,
\end{equation}
where each diagonal block acts on an entry of $[\alpha ,\beta ,C]$, and $\hat G_A, \hat G_B \in \hat \g_2$. In Subsection \ref{rev-trans-gbit} it is argued that $\hat \g_2 \subseteq {\rm O}(d_2)$, hence local transformations (\ref{3-action}) are orthogonal.

Lemma~\ref{l.SGS-1} shows the existence of a real matrix $S>0$ such that for any $\hat G\in \hat \g_{2,2}$ the matrix $S\hat G S^{-1}$ is orthogonal. In particular 
\[
	\big[S\cdot (G_A \otimes G_B)^\wedge\cdot S^{-1} \big]\t \big[S\cdot (G_A \otimes G_B)^\wedge \cdot S^{-1} \big] =\id\ ,
\]
which implies the commutation relation
\begin{equation}\label{LS=SL}
	S\cdot (G_A \otimes G_B)^\wedge = (G_A \otimes G_B)^\wedge \cdot S
\end{equation}
Subsection \ref{rev-trans-gbit} concludes that $d_2$ is odd, and that $\mathrm{SO} (d_2) \subseteq \hat \g_2$ except when $d_2=7$, where $M\mathbb{G}_{2} M^{-1} \subseteq \hat \g_2$ and $\mathbb{G}_{2}$ is the fundamental representation of the smallest exceptional Lie group \cite{group_book}. For $d_2 \geq 3$ these groups act irreducibly in $\mathbb{C}^{d_2}$, hence $\hat \g_2$ acts irreducibly in $\mathbb{C}^{d_2}$ too \cite{group_book}. The first two diagonal blocks in (\ref{3-action}) are irreducible. The exterior tensor product of two irreducible representations (in $\mathbb{C}^{d}$) is also an irreducible representation, hence the third diagonal block in (\ref{3-action}) is also irreducible. This together with (\ref{LS=SL}) implies that
\[
	S=\left(\begin{array}{ccc}
		a\id &\z &\z \\
		\z & b\id &\z \\	
		\z &\z & s\id
	\end{array}\right)
\]
for some $a,b,s > 0$ (Schur's Lemma~\cite{group_book}).

According to Lemma~\ref{PseudoSwap}, for each unit vector $\alpha \in \hat \s_2$ there is a transformation $G_\mathrm{swap} \in \g_{2,2}$ such that
\begin{eqnarray*}
	&& \hat G_\mathrm{swap} [\alpha, 0, 0] 
	\\ &=& 
	\hat G_\mathrm{swap} \left( [\alpha, \alpha, \alpha\alpha\t] +[\alpha, -\alpha, -\alpha\alpha\t] \right)/2 
	\\ &=& 
	\left( [\alpha, \alpha, \alpha\alpha\t] + [-\alpha, \alpha, -\alpha\alpha\t] \right)/2
	\\ &=& 
	[0, \alpha, 0]\ .
\end{eqnarray*}
Since $S\hat G_\mathrm{swap} S^{-1}$ is orthogonal, the vectors $[\alpha,0,0]$ and $[0,b a^{-1} \alpha,0]$ have the same modulus, hence $a=b$.
Also, there is a transformation $G_\mathrm{cnot} \in \g_{2,2}$ such that
\begin{eqnarray*}
	&&\hat G_\mathrm{cnot} [0,0,\alpha\alpha\t] 
	\\ &=& 
	\hat G_\mathrm{cnot} \left( [\alpha, \alpha, \alpha\alpha\t] +[-\alpha, -\alpha, \alpha\alpha\t] \right)/2 
	\\ &=& 
	\left( [\alpha, \alpha, \alpha\alpha\t] + [-\alpha, \alpha, -\alpha\alpha\t] \right)/2
	\\ &=& 
	[0, \alpha, 0]\ .
\end{eqnarray*}
Since $S\hat G_\mathrm{cnot} S^{-1}$ is orthogonal, the vectors $[0,0, \alpha\alpha\t]$ and $[0,b s^{-1} \alpha,0]$ have the same modulus, hence $s=b$. Consequently $S=a\id$, and the claim follows.
\end{proof}

\begin{lemma}
\label{LemFive}
The results of this paper also hold if Requirement~\ref{a.effects} is replaced by Requirement 5'.
\end{lemma}
\proof
Assume Requirement 5', but not Requirement~\ref{a.effects}. It follows that
$\s_1$ contains a single state only: if it contained more than one state, there would exist some $\psi\in\s_1$
which is not completely mixed, which would then be distinguishable from some other state, contradicting that $\s_1$ has capacity $1$.
Requirement~\ref{a.effects} is used in the proof of Theorem~\ref{pure=boundary}.
This proof is easily modified to comply with Requirement 5' instead: adopting the notation from the proof,
the state $\varphi_\mathrm{mix}$ is not completely mixed, thus distinguishable from some other state. This proves existence of some $\hat\Omega$
with the claimed properties, and the other arguments remain unchanged, proving that $\hat \s_2$ can be represented as a unit ball.
All pure states in $\hat\s_2$ are not completely mixed, hence have a corresponding tight effect which is physically allowed.
But for every state on the surface of the ball, there exists only one unique tight effect which gives probability one for that state.
Hence, all these effects must be allowed, and since they generate the set of all effects, this proves Requirement~\ref{a.effects}
for use in~(\ref{thetai}) and the rest of the paper.
\qed


\begin{thebibliography}{99}

\bibitem{J.Bell} J.\ S.\ Bell, Physics {\bf 1}, 195 (1964).

\bibitem{n-factoring} P.\ W.\ Shor, {\em Polynomial-Time Algorithms for Prime Factorization and Discrete Logarithms on a Quantum Computer}; SIAM J.\ Comput.\ {\bf 26}(5), 1484-1509 (1997), quant-ph/9508027v2.

\bibitem{tomo} L.\ Hardy, {\em Foliable Operational Structures for General Probabilistic Theories}; arXiv:0912.4740v1.

\bibitem{Hardy5} L.\ Hardy; {\em Quantum Theory From Five Reasonable Axioms}, quant-ph/0101012v4.

\bibitem{Barn} H.\ Barnum, A.\ Wilce, {\em Information processing in convex operational theories}, DCM/QPL (Oxford University 2008), arXiv:0908.2352v1.

\bibitem{Barr} J.\ Barrett, {\em Information processing in generalized probabilistic theories}, Phys.\ Rev.\ A {\bf 75}, 032304 (2007),
arXiv:quant-ph/0508211v3.

\bibitem{gwmackey} G.\ W.\ Mackey; {\em The mathematical foundations of quantum mechanics}, (W.\ A.\ Benjamin Inc, New York, 1963).

\bibitem{BvN} G.\ Birkhoff, J.\ von Neumann, {\em The Logic of Quantum Mechanics}, Annals of Mathematics, {\bf 37}, 823 (1936).

\bibitem{darian} G.\ Chiribella, G.\ M.\ D'Ariano, P.\ Perinotti; {\em Probabilistic theories with purification}; Phys. Rev. A {\bf 81}, 062348 (2010), arXiv:0908.1583v5.

\bibitem{vandam} W.\ van Dam, {\em Implausible Consequences of Superstrong Nonlocality}, arXiv:quant-ph/0501159v1.

\bibitem{ic} M.\ Pawlowski, T.\ Paterek, D.\ Kaszlikowski, V.\ Scarani, A.\ Winter, M.\ Zukowski, {\em A new physical principle: Information Causality}, Nature {\bf 461}, 1101 (2009), arXiv:0905.2292v3.

\bibitem{pr} S.\ Popescu, D.\ Rohrlich, {\em Causality and Nonlocality as Axioms for Quantum Mechanics},
Proceedings of the Symposium on Causality and Locality in Modern Physics and Astronomy (York University, Toronto, 1997), arXiv:quant-ph/9709026v2.

\bibitem{nav} M.\ Navascues, H.\ Wunderlich, {\em A glance beyond the quantum model}, Proc.\ Roy.\ Soc.\ Lond.\ A {\bf 466}, 881-890 (2009), arXiv:0907.0372v1.

\bibitem{boxworld}
D.\ Gross, M.\ M\"uller, R.\ Colbeck, O.\ C.\ O.\ Dahlsten, \emph{All reversible dynamics in maximally non-local
theories are trivial}, Phys.\ Rev.\ Lett.\ {\bf 104}, 080402 (2010), arXiv:0910.1840v2.

\bibitem{Gleason} A. Gleason, J. Math. Mech. 6, 885 (1957).

\bibitem{Weinberg} S. Weinberg, Ann. Phys. NY {\bf 194}, 336 (1989).

\bibitem{Gisin} N.\ Gisin, \emph{Weinberg's non-linear quantum mechanics and supraluminal communications}, Phys.\ Lett.\ A {\bf 143}1--2 (1990).

\bibitem{theoryspace} S.\ Aaronson, {\em Is Quantum Mechanics An Island In Theoryspace?}, quant-ph/0401062v2.

\bibitem{Daki} B.\ Daki\'c, C.\ Brukner, {\em Quantum Theory and Beyond: Is Entanglement Special?}, arXiv:0911.0695v1.

\bibitem{bFuchs} C.\ A.\ Fuchs, {\em Quantum Mechanics as Quantum Information (and only a little more)},
Quantum Theory: Reconstruction of Foundations, A.\ Khrenikov (ed.), V\"axjo University Press (2002), arXiv:quant-ph/0205039v1.

\bibitem{br} G.\ Brassard, {\em Is information the key?} Nature Physics {\bf 1}, 2 (2005).

\bibitem{AlfsenShultzBook}
E.\ M.\ Alfsen and F.\ W.\ Shultz, {\em Geometry of state spaces of operator algebras}, Birkh\"auser, Boston (2003).

\bibitem{convex_book} R.\ T.\ Rockafellar, {\em Convex Analysis}, Princeton University Press (1970).

\bibitem{matrix_groups} A.\ Baker, {\em Matrix Groups, An Introduction to Lie Group Theory}, Springer-Verlag London Limited (2006).

\bibitem{PCShort} A.\ J.\ Short, {\em private communication.}

\bibitem{Haar-book} B.\ Simon, {\em Representations of Finite and Compact Groups}, Graduate Studies in Mathematics, vol.\ 10, American Mathematical Society (1996).

\bibitem{Strask} S.\ Straszewicz, \emph{\"Uber exponierte Punkte abgeschlossener Punktmengen}, Fund.\ Math.\ {\bf 24}, 139-143 (1935).

\bibitem{Onishchik1}
A.\ L.\ Onishchik and V.\ V.\ Gorbatsevich, \emph{Lie groups and Lie algebras I}, Encyclopedia of Mathematical Sciences
20, Springer Verlag Berlin, Heidelberg (1993).

\bibitem{Onishchik2}
A.\ L.\ Onishchik, \emph{Transitive compact transformation groups}, Mat.\ Sb.\ (N.S.) {\bf 60}(102):4 447--485 (1963);
English translation: Amer.\ Math.\ Soc.\ Transl.\ (2) {\bf 55}, 153--194 (1966).

\bibitem{group_book} W.\ Fulton, J.\ Harris, {\em Representation Theory}, Graduate texts in mathematics, Springer (2004).

\bibitem{Bargmann}
V.\ Bargmann, \emph{Note on Wigner's Theorem on Symmetry Operations}, J.\ Math.\ Phys.\ {\bf 5}, 862--868 (1964).

\bibitem{Wigner}
E.\ P.\ Wigner, \emph{Normal Form of Antiunitary Operators}, J.\ Math.\ Phys.\ {\bf 1}, 409--413 (1960).

\bibitem{Bremner}
M.\ J.\ Bremner, C.\ M.\ Dawson, J.\ L.\ Dodd, A.\ Gilchrist, A.\ W.\ Harrow, D.\ Mortimer, M.\ A.\ Nielsen, T.\ J.\ Osborne,
\emph{Practical scheme for quantum computation with any two-qubit entangling gate}, Phys.\ Rev.\ Lett.\ {\bf 89}:247902 (2002),
arXiv:quant-ph/0207072v1.

\bibitem{Giulio}
G.\ Chiribella, G.\ M.\ D'Ariano, and P.\ Perinotti, \emph{Informational derivation of Quantum Theory}, arXiv:1011.6451v2.

\end{thebibliography}
\end{document}